\documentclass[a4paper,UKenglish,cleveref, autoref, thm-restate]{lipics-v2021}

\usepackage[boxruled,vlined,linesnumbered]{algorithm2e}
\usepackage{kbordermatrix}
\usepackage{apxproof}
\usepackage{tikz}
\usetikzlibrary{arrows}

\usepackage{amsmath}
\usepackage{thmtools}
\usepackage{comment}
\usepackage{thm-restate}
\usepackage{xspace}
\usepackage{dsfont}
\usepackage{multirow}
\usepackage{booktabs}
\usepackage[colorinlistoftodos,prependcaption,textsize=tiny]{todonotes}
\usepackage{wrapfig}

\title{Branch Prediction Analysis of \MPLong and \KMPLong Algorithms} %

\titlerunning{Branch Prediction Analysis of \MPLong and \KMPLong Algorithms}

\author{Cyril Nicaud}{Univ Gustave Eiffel, CNRS, LIGM, F-77454 Marne-la-Vallée, France}{cyril.nicaud@univ-eiffel.fr}{}{}

\author{Carine Pivoteau}{Univ Gustave Eiffel, CNRS, LIGM, F-77454 Marne-la-Vallée, France}{carine.pivoteau@univ-eiffel.fr}{}{}

\author{Stéphane Vialette}{Univ Gustave Eiffel, CNRS, LIGM, F-77454 Marne-la-Vallée, France}{stephane.vialette@univ-eiffel.fr}{}{}

\authorrunning{C. Nicaud, C. Pivoteau and S. Vialette} %TODO mandatory. First: Use abbreviated first/middle names. Second (only in severe cases): Use first author plus 'et al.'

\Copyright{Cyril Nicaud, Carine Pivoteau and Stéphane Vialette} %TODO mandatory, please use full first names. LIPIcs license is "CC-BY";  http://creativecommons.org/licenses/by/3.0/


 %TODO mandatory; please add comma-separated list of keywords

\category{} %optional, e.g. invited paper

\relatedversion{} %optional, e.g. full version hosted on arXiv, HAL, or other respository/website
\nolinenumbers %uncomment to disable line numbering

\EventEditors{}
\EventNoEds{2}
\EventLongTitle{submitted}%36th Annual Symposium on Combinatorial Pattern Matching (CPM 2025)}
\EventShortTitle{submitted}
\EventAcronym{submitted}
\EventYear{}
\EventDate{}
\EventLocation{}
\EventLogo{}
\SeriesVolume{}
\ArticleNo{}
\newcommand{\sNN}{\underline\nu}
\newcommand{\sN}{\nu}
\newcommand{\sT}{\tau}
\newcommand{\sTT}{\underline\tau}

\newcommand{\StronglyTaken}{\textsc{Strongly taken}\xspace}
\newcommand{\WeaklyTaken}{\textsc{Weakly taken}\xspace}
\newcommand{\StronglyNotTaken}{\textsc{Strongly not taken}\xspace}
\newcommand{\WeaklyNotTaken}{\textsc{Weakly not taken}\xspace}

\newcommand{\MP}{{MP}\xspace}
\newcommand{\KMP}{{KMP}\xspace}
\newcommand{\MPLong}{Morris-Pratt\xspace}
\newcommand{\KMPLong}{Knuth-Morris-Pratt\xspace}

\newcommand{\Text}{W}
\newcommand{\Pattern}{X}
\newcommand{\alphabet}{A}
\newcommand{\anyletter}{\alpha}

\newcommand{\A}{\mathcal{A}}
\newcommand{\T}{\mathcal{T}}

\newcommand{\fmp}{\textrm{mp}}

\newcommand{\MX}{{\mathcal{F}^\fmp_\Pattern}}
\newcommand{\MTX}{{\T^\fmp_\Pattern}}
\newcommand{\WMTX}{{\widetilde{\T}^\fmp_\Pattern}}
\newcommand{\PX}{{\mathcal{P}^\fmp_\Pattern}}

\newcommand{\fkmp}{\textrm{kmp}}

\newcommand{\KX}{{\mathcal{F}^\fkmp_\Pattern}}
\newcommand{\KTX}{{\T^\fkmp_\Pattern}}
\newcommand{\WKTX}{{\widetilde{\T}^\fkmp_\Pattern}}
\newcommand{\PKX}{{\mathcal{P}^\fkmp_\Pattern}}

\renewcommand{\O}{\mathcal{O}}

\newcommand{\border}[1]{\textrm{bord}(#1)}

\DeclareMathOperator{\pref}{Pref}
\DeclareMathOperator{\suff}{Suff}

\DeclareMathOperator{\OutputTransducer}{\nabla}

\begin{document}

\maketitle

\begin{abstract}
    We analyze the classical Morris-Pratt and Knuth-Morris-Pratt pattern matching algorithms through the lens of computer architecture, investigating the impact of incorporating a simple branch prediction mechanism into the model of computation. Assuming a fixed pattern and a random text, we derive precise estimates of the number of mispredictions these algorithms produce using local predictors. Our approach is based on automata theory and Markov chains, providing a foundation for the theoretical analysis of other text algorithms and more advanced branch prediction strategies.  
\end{abstract}

\newpage
\section{Introduction}\label{sec:introduction}

Pipelining is a fundamental technique employed in modern processors to improve performance by executing multiple instructions in parallel, rather than sequentially. Instructions execution is divided into distinct stages, enabling the simultaneous processing of multiple instructions, much like an assembly line in a factory. Today, pipelining is ubiquitous, even in low-cost processors priced under a dollar.
We refer the reader to the textbook reference~\cite{HePa17} for more details on this important architectural feature.

The sequential execution model assumes that each instruction completes before the next one begins; however, this assumption does not hold in a pipelined processor. Specific conditions, known as hazards, may prevent the next instruction in the pipeline from executing during its designated clock cycle. Hazards introduce delays that undermine the performance benefits of pipelining and may stall the pipeline, thus reducing the theoretical speedup:
\begin{itemize}
    \item \emph{Structural hazards} occur when resource conflicts arise, preventing the hardware from supporting all possible combinations of instructions during simultaneous executions.
    \item \emph{Control hazards} arise from the pipelining of jumps and other instructions that modify the order in which instructions are processed, by updating the Program Counter (PC).
    \item \emph{Data hazards} occur when an instruction depends on the result of a previous instruction, and this dependency is exposed by the overlap of instructions in the pipeline.
\end{itemize}
To minimize stalls due to control hazards and improve execution efficiency, modern computer architectures incorporate \emph{branch predictors}. A branch predictor is a digital circuit that anticipates the outcome of a branch (e.g., an \texttt{if–then–else} structure) before it is determined.

Two-way branching is typically implemented using a conditional jump instruction, which can either be \emph{taken}, updating the Program Counter to the target address specified by the jump instruction and redirecting  the execution path to a different location in memory, or \emph{not taken}, allowing execution to continue sequentially. The outcome of a conditional jump remains unknown until the condition is evaluated and the instruction reaches the actual execution stage in the pipeline. Without branch prediction, the processor would be forced to wait until the conditional jump instruction reaches this stage before the next instruction can enter the first stage in the pipeline. The branch predictor aims to reduce this delay by predicting whether the conditional jump is likely to be taken or not. The instruction corresponding to the predicted branch is then fetched and speculatively executed. If the prediction is later found to be incorrect (i.e., a \emph{misprediction} occurs), the speculatively executed or partially executed instructions are discarded, and the pipeline is flushed and restarted with the correct branch, incurring a small delay. 

The effectiveness of a branch prediction scheme depends on both its accuracy and the frequency of conditional branches. 
\emph{Static branch prediction} is the simplest technique, as it does not depend on the code execution history and, therefore, cannot adapt to program behavior.
In contrast, \emph{dynamic branch prediction} takes advantage of runtime information (specifically, branch execution history) to determine whether branches were taken or not, allowing it to make more informed predictions about future branches.

A vast body of research is dedicated to dynamic branch prediction schemes. At the highest level, branch predictors are classified into two categories: global and local.
A \emph{global branch predictor} does not maintain separate history records for individual conditional jumps. Instead, it relies on a shared history of all jumps, allowing it to capture their correlations and improve prediction accuracy.
In contrast, a \emph{local branch predictor} maintains an independent history buffer for each conditional jump, enabling predictions based solely on the behavior of that specific branch.
Since the 2000s, modern processors typically employ a combination of local and global branch prediction techniques, often incorporating even more sophisticated designs. For a deeper exploration of this topic, see~\cite{Mittal2018}, and for a comprehensive discussion of modern computer architecture, refer to~\cite{HePa17}.

In this study, we focus on local branch predictors implemented with \emph{saturated counters}. A \emph{1-bit saturating counter} (essentially a flip-flop) records the most recent branch outcome. Although this is the simplest form of dynamic branch prediction, it offers limited accuracy. A \emph{2-bit saturating counter} (see \cref{fig:predictor}), by contrast, operates as a state machine with four possible states: \StronglyNotTaken ($\sNN$), \WeaklyNotTaken ($\sN$), \WeaklyTaken ($\sT$), and \StronglyTaken ($\sTT$). When the 2-bit saturated branch predictor is in the \StronglyNotTaken or \WeaklyNotTaken state, it predicts that the branch will not be taken and  execution will proceed sequentially. Conversely, when the predictor is in the \StronglyTaken or \WeaklyTaken state, it predicts that the branch will be taken, meaning execution will jump to the target address. 
Each time a branch is evaluated, the corresponding state machine updates its state. If the branch is not taken, the state shifts toward \StronglyNotTaken; if taken, it moves toward \StronglyTaken.  
A misprediction (corresponding to a bold edge in  \cref{fig:predictor}) occurs when:  
\begin{itemize}
    \item a branch is not taken, while the predictor is in either of the \textsc{Taken} states ($\sT$ or $\sTT$); 
    \item a branch is taken, while the predictor is in either of the \textsc{Not Taken} states ($\sN$ or $\sNN$).
\end{itemize}  
This mechanism gives the 2-bit saturating counter an advantage over the simpler 1-bit scheme: a branch must deviate twice from its usual behavior (i.e. a \textsc{Strongly} state) before the prediction changes, reducing the likelihood of mispredictions.

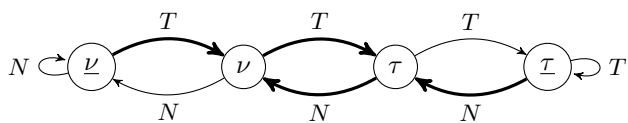
\begin{figure}[t]
    \centering
    \begin{tikzpicture}
        \node[draw,circle] (NN) at (0,0) {$\sNN$};
        \node[draw,circle] (N) at (2,0) {$\sN$};
        \node[draw,circle] (T) at (4,0) {$\sT$};
        \node[draw,circle] (TT) at (6,0) {$\sTT$};

        \draw (NN) edge[> = stealth',loop left] node {\small $N$} (NN);
        \draw (N) edge[> = stealth',->,bend left] node[below] {\small $N$} (NN);
        \draw (T) edge[> = stealth',->,bend left, very thick] node[below] {\small $N$} (N);
        \draw (TT) edge[> = stealth',->,bend left, very thick] node[below] {\small $N$} (T);
       
        \draw (NN) edge[> = stealth',->,bend left, very thick] node[above] {\small $T$} (N);
        \draw (N) edge[> = stealth',->,bend left, very thick] node[above] {\small $T$} (T);
        \draw (T) edge[> = stealth',->,bend left] node[above] {\small $T$} (TT);
        \draw (TT) edge[> = stealth',loop right] node {\small $T$} (TT);

    \end{tikzpicture}
    \caption{\label{fig:predictor}%
        The 2-bit saturated predictor consists of four states: $\sNN$ and $\sN$ predict that the branch will not be taken, while $\sT$ and $\sTT$ predict that it will. The predictor updates at each condition evaluation, transitioning via $T$ when the branch is taken (i.e., the condition is true) and via $N$ when it is not. Bold edges indicate mispredictions.
    }
\end{figure}

This paper provides an initial theoretical exploration of the techniques and results involved in analyzing pattern matching algorithms within computational models enhanced by branch prediction mechanisms. Specifically, we examine the classical Morris-Pratt~(\MP) and Knuth-Morris-Pratt~(\KMP) algorithms~\cite{MP,knuth1977fast} in such models, with a primary focus on quantifying mispredictions for random text inputs. Over the past two decades, research in this area has progressed from experimental studies on sorting algorithms and fundamental data structures~\cite{BrMo05} to theoretical analyses of misprediction behavior in \textsc{Java}'s double-pivot Quicksort~\cite{MaNeWi15} and the development of skewed variants of classical algorithms (e.g., binary search, exponentiation by squaring) designed to optimize branch prediction~\cite{AuNiPi2016}. These studies have primarily examined local predictors, particularly 2-bit saturating ones (see \cref{fig:predictor}), which also serve as the foundation of this work.
 
\section{Algorithms and their encoding using transducers}\label{sec:algorithms} 

Throughout the article, indices start at $0$: if $u$ is a word of length $|u|=n$ over the alphabet~$\alphabet$, we represent it as $u=u_0 \dots u_{n-1}$, where each $u_i$ is a letter of $u$. 
We also use $u[i]$ to denote the letter $u_i$. 
If $u=xyz$ where $x$, $y$ and $z$ are words, then $x$ is a \emph{prefix} of $u$, $y$ is a \emph{factor} and $z$ is a \emph{suffix}. A prefix (resp. suffix) of $u$ is \emph{strict} if it differs from $u$. 
For any $i \in \{0, \dots, n\}$, we denote by  $\pref(u,i)$  the prefix  of $u$ of length $i$. 
A \emph{strict border} of $u$ is a word $v$ that is both a strict prefix and a strict suffix of $u$.

\vspace{-.3cm}
\paragraph*{Algorithms \MP and \KMP}\label{sec:MP KMP}
The Morris-Pratt (\MP) and Knuth-Morris-Pratt (\KMP) algorithms are textbook solutions to the pattern matching problem~\cite{DBLP:books/ox/CrochemoreR94, DBLP:books/cu/Gusfield1997, DBLP:books/daglib/0020103}. 
Both rely on precomputing a \emph{failure function}, which helps identify candidate positions for the pattern $\Pattern$ in the text $\Text$. 
The general approach involves scanning $\Text$ from left to right, one letter at a time. 
Before moving to the next letter in~$\Text$, the algorithm determines the longest prefix of $\Pattern$ that is also a suffix of the discovered prefix of $\Text$.
The failure function allows this computation to be performed efficiently.

The function $\fmp_\Pattern$  maps each prefix of $\Pattern$ to its longest border, with the convention that $\fmp_\Pattern(\varepsilon)=\bot$. 
The function $\fkmp_\Pattern$ is a refinement of $\fmp_\Pattern$ defined by 
$\fkmp_\Pattern(\Pattern) = \fmp_\Pattern(\Pattern)$, $\fkmp_\Pattern(\varepsilon)=\bot$, and for all prefix $u \anyletter$ of $\Pattern$, where $u\in\alphabet^*$ and $\anyletter\in\alphabet$, $\fkmp_\Pattern(u)$ is the longest strict suffix of $u$ that is also a prefix of $u$ but such that $\fkmp_\Pattern(u)\alpha$ is not.
If  no such strict suffix exists, then $\fkmp_\Pattern(u)=\bot$. 
See~\cite{DBLP:books/cu/Gusfield1997} for a more detailed discussion of these failure functions. 
In the following, we only require that Algorithm~\ref{algo:find} remains correct regardless of which failure function is used.
Within the algorithm, the function $b:=\fmp_\Pattern$ (or $b:=\fkmp_\Pattern$) is transformed into a precomputed integer-valued array $B$, defined as  
$B[i] = |b(\pref(\Pattern,i))|$ for $i \in\{0,\dots,|\Pattern|\}$, with the convention that $|\bot| = -1$.

\SetProcNameSty{textsc}
\SetProcArgSty{textsc}

\begin{wrapfigure}{l}{0.48\textwidth}
    %\noindent
    \begin{minipage}{.48\textwidth}
        \small
        \begin{algorithm}[H]
            \DontPrintSemicolon
            $m,n\gets |\Pattern|, |\Text|$\;
            $i,j,nb\gets 0,0,0$\;
            \While{$j<n$\label{line:find-main-while}}
            {
               \While{$i\geq 0$ \bf{and} $\Pattern[i]\neq \Text[j]$ \label{line:find-nested-while}}
                {
                $i\gets B[i]$\;
                }
                $i,j\gets i+1,j+1$\;
                \If{$i=m$}{\label{line:find-counter-if}
                $i\gets B[i]$\;
                $nb \gets nb+1$\;
                }
            }
            \Return $nb$
            %\SetAlgoRefName{\scshape Find}
            \SetAlgoRefName{FIND}
            \caption{\label{algo:find}%
                \!($\Pattern$, $\Text$, $B$
            )}
        \end{algorithm}
    \end{minipage}
\end{wrapfigure}

Algorithm~\ref{algo:find}, shown on the left, utilizes the precomputed table $B$ from either $\fmp_\Pattern$ or $\fkmp_\Pattern$ to efficiently locate potential occurrences of $\Pattern $ in $\Text$.  
% \carine{Algorithm~\ref{algo:find}, shown the left, utilizes the table $B$ to efficiently locate potential occurrences of $\Pattern $ in $\Text$.
The indices $i$ and $j$ represent the current positions in $\Pattern$ and $\Text$, respectively.
The main {\tt while} loop (Line~\ref{line:find-main-while}) iterates once for each letter discovered in $\Text$. At the start of each iteration, index $i$ holds the length of the longest matching prefix of $\Pattern$. The inner {\tt while} loop (Line~\ref{line:find-nested-while}) updates $i$ using the precomputed table $B$. Finally, the {\tt if} statement (Line~\ref{line:find-counter-if}) is triggered when an occurrence of $\Pattern$ is found, updating $i$ accordingly.  
For both \MP and \KMP, the table $B$ can be computed in~$\O(m)$ time and
\ref{algo:find} runs in $\O(n)$ time. 
More precisely, in the worst case, Algorithm~\ref{algo:find} performs at most $2n-m$ letter comparisons~\cite{DBLP:books/ox/CrochemoreR94}.

Note that in any programming language supporting short-circuit evaluation of Boolean operators, the condition $i \geq 0$ {\tt and} $\Pattern[i] \neq \Text[j]$ at Line~\ref{line:find-nested-while} of Algorithm~\ref{algo:find} is evaluated as two separate jumps by the compiler.  
As a result, Algorithm~\ref{algo:find} contains a total of four branches: one at Line~\ref{line:find-main-while}, two at Line~\ref{line:find-nested-while}, and one at Line~\ref{line:find-counter-if}. In our model, each of these four branches is assigned a local predictor, and all may potentially lead to mispredictions. 
When a conditional instruction is compiled, it results in a jump that can correspond to either a taken or a not-taken branch.\footnote{Most assembly languages provide both {\tt jump-if-true} and {\tt jump-if-false} instructions.} For consistency, we define a successful condition --~when the test evaluates to true~-- as always leading to a taken branch. This convention does not affect our analysis, as the predictors we consider are symmetric.

\vspace{-.3cm}\paragraph*{Associated automata}\label{sec:transducers}

At the beginning of each iteration of the main \texttt{while} loop of Algorithm \textsc{Find}, the prefix of length $j$ of $\Text$ (i.e. from letter $\Text[0]$ to letter $\Text[j-1]$) has been discovered, and $i$ is the length of the longest suffix of $\pref(\Text,j)$ that is a strict prefix of $\Pattern$: we cannot have $i=m$, because when it happens, the pattern is found and $i$ is directly updated to $B[i]$, Line~\ref{line:find-counter-if}.

The evolution of $i$ at each iteration of the main {\tt while} loop is encoded by a deterministic and complete automaton $\A_\Pattern$. Its set of states is the set $Q_\Pattern$
of strict prefixes of $\Pattern$, identified by their unique lengths if necessary. Its transition function $\delta_\Pattern$ maps a pair $(u,\anyletter)$ to the longest suffix of $u\anyletter$ which is in $Q_\Pattern$. Its initial state is $\varepsilon$. If $\Pattern=Y\anyletter$, where $\anyletter\in\alphabet$ is a letter, then when following the path labeled by $\Text$ starting from the initial state of $\A_\Pattern$, there is an occurrence of $\Pattern$ in $\Text$ exactly when the transition $Y\xrightarrow{\anyletter}\delta_\Pattern(Y,\alpha)$ is used. This variant of the classical construction is more relevant for this article than the usual one~\cite[Sec. 7.1]{DBLP:books/ox/CrochemoreR94} which also has the state $\Pattern$. The transition $Y\xrightarrow{\anyletter}\delta_\Pattern(Y,\alpha)$ is the accepting transition to identify occurrences of $\Pattern$. The automaton $\A_\Pattern$ tracks the value of $i$ at the beginning of each iteration of the main loop, where $i$ corresponds to the length of the current state label. This value remains the same for both \MP and \KMP. An example of $\A_\Pattern$ is depicted in \cref{fig:dfa}.

\begin{figure}[t]
	\centering
    \small
	\begin{tikzpicture}[scale=1]
		\node[draw,thick] (0) at (0,0) {$\varepsilon$};
		\node[draw,thick] (1) at (2,0) {$a$};
		\node[draw,thick] (2) at (4,0) {$ab$};
		\node[draw,thick] (3) at (6,0) {$aba$};
		\node[draw,thick] (4) at (8,0) {$abab$};
		
		\draw[->] (-.4,0) -- (0);
		
		\draw (0) edge[> = stealth',->, shorten >= .1em] node[above]{$a$} (1);
		\draw (0) edge[> = stealth',loop above] node[above,left, xshift=-1mm]{$b$} (0);
		\draw (1) edge[> = stealth',->, shorten >= .1em] node[above]{$b$} (2);
		\draw (1) edge[> = stealth',loop above] node[above,left, xshift=-1mm]{$a$} (1);
		\draw (2) edge[> = stealth', ->, shorten >= .1em] node[above,yshift=0.2mm]{$a$} (3);
		\draw (2) edge[> = stealth', ->, bend left, shorten >= .1em] node[above]{$b$} (0);
		\draw (3) edge[> = stealth', ->, bend right, shorten >= .1em] node[above,yshift=0.2mm]{$a$} (1);
		\draw (3) edge[> = stealth', ->, shorten >= .1em] node[below]{$b$} (4);
	    \draw (4) edge[> = stealth', ->, bend right, shorten >= .1em] node[above,yshift=0.2mm]{$a$} (3);
		\draw (4) edge[> = stealth', ->, bend left, shorten >= .1em, very thick, green!50!black] node[above]{$b$} (0);
		
	\end{tikzpicture}
	\caption{The deterministic and complete automaton $\A_\Pattern$ for $\Pattern=ababb$.\label{fig:dfa}}
\end{figure}

To refine the simulation of the \textsc{Find} algorithm using automata, we incorporate the failure functions. This is achieved by constructing the \emph{failure automaton}. Specifically, for Algorithm~\MP, let $\MX$ be the automaton defined by:  
\begin{itemize}
    \item A state set $Q_\Pattern \cup \{\bot\}$ and an initial state $\varepsilon$.  
    \item Transitions $\bot \xrightarrow{\anyletter} \varepsilon$ for every $\anyletter \in \alphabet$.  
    \item Transitions $u \xrightarrow{\anyletter} u\anyletter$ for every $u \in Q_\Pattern$ such that $u\anyletter \in Q_\Pattern$.  
    \item A failure transition $u \rightarrow \fmp_\Pattern(u)$ for every $u \in Q_\Pattern$, used when attempting to read a letter $\anyletter$ where $u\anyletter \notin Q_\Pattern$.  
\end{itemize}  
The automaton $\KX$ associated with \KMP is identical to $\MX$, except that its failure transitions are defined as $u \rightarrow \fkmp_\Pattern(u)$ for every $u \in Q_\Pattern$. Both automata serve as graphical representations of the failure functions $\fmp_\Pattern$ and $\fkmp_\Pattern$, structured in a way that aligns with~$\A_\Pattern$. An example of $\MX$ and $\KX$ is illustrated in \cref{fig:dfa failure}.  

\begin{figure}[t]
	\centering
    \small
	\begin{tikzpicture}[scale=1]
        \node[draw,thick] (-1) at (-2,0) {\small $\bot$};
        \node[draw,thick] (0) at (0,0) {$\varepsilon$};
        \node[draw,thick] (1) at (2,0) {$a$};
        \node[draw,thick] (2) at (4,0) {$ab$};
        \node[draw,thick] (3) at (6,0) {$aba$};
        \node[draw,thick] (4) at (8,0) {$abab$};
        
        \draw[->] (0,-.4) -- (0);
        \draw (-1) edge[> = stealth',->, shorten >= .1em, purple] node[fill=white]{$a,b$} (0);
        
        \draw (0) edge[> = stealth',->, shorten >= .1em] node[above]{$a$} (1);
        \draw (1) edge[> = stealth',->, shorten >= .1em] node[above]{$b$} (2);
        \draw (2) edge[> = stealth', ->, shorten >= .1em] node[above,yshift=0.2mm]{$a$} (3);
        \draw (3) edge[> = stealth', ->, shorten >= .1em] node[above]{$b$} (4);
        
        % MP
        \draw (0) edge[> = stealth',->, shorten >= .1em, bend right=40, red, dotted, thick] (-1);
        \draw (1) edge[> = stealth',->, shorten >= .1em, bend right=40, red, dotted, thick] (0);
        \draw (2) edge[> = stealth',->, shorten >= .1em, bend right=40, red, dotted, thick] (0);
        \draw (3) edge[> = stealth',->, shorten >= .1em, bend right, red, dotted, thick] (1);
        \draw (4) edge[> = stealth',->, shorten >= .1em, bend right, red, dotted, thick] (2);
	    
        % KMP
        \draw (0) edge[> = stealth',->, shorten >= .1em, bend left=40, blue, dashed] (-1);
        \draw (1) edge[> = stealth',->, shorten >= .1em, bend left=40, blue, dashed] (0);
        \draw (2) edge[> = stealth',->, shorten >= .1em, bend left=40, blue, dashed] (-1);
        \draw (3) edge[> = stealth',->, shorten >= .1em, bend left=40, blue, dashed] (0);
        \draw (4) edge[> = stealth',->, shorten >= .1em, bend left, blue, dashed] (2);
        
	\end{tikzpicture}
	\caption{The automata $\MX$ and $\KX$ for $\Pattern=ababb$, on the same picture; the failure transitions of $\MX$ are in dotted red lines and above, those of $\KX$ are in dashed blue lines and below. To read the letter $a$ from state $aba$ in $\MX$, one follows the failure transition $aba\rightarrow a$ then $a\rightarrow\varepsilon$ until one can finally read  $\varepsilon\xrightarrow{a}a$. In $\KX$, only one failure transition $aba\rightarrow\varepsilon$ is needed, instead of two.\label{fig:dfa failure}}
\end{figure}

When reading a letter $\anyletter$ from a state $u$ in $\MX$ or $\KX$, if the transition $u \xrightarrow{\anyletter} u\anyletter$ does not exist, the automaton follows failure transitions until a state with an outgoing transition labeled by $\anyletter$ is found. This process corresponds to a single backward transition in $\A_\Pattern$.  
Crucially, using a failure transition directly mirrors the execution of the nested \texttt{while} loop in Algorithm \textsc{Find} (Line~\ref{line:find-nested-while}) or triggers the \texttt{if} statement at Line~\ref{line:find-counter-if} when an occurrence of $\Pattern$ is found. This construction captures what we need for the forthcoming analysis.

\section{Expected number of letter comparisons for a given pattern}\label{sec:comparisons}  

In this section, we refine the analysis of the expected number of letter comparisons performed by Algorithm~\ref{algo:find} for a given pattern $\Pattern$ of length $m$ and a random text $\Text$ of length $n$. The average-case complexity of classical pattern matching algorithms has been  explored before, particularly in scenarios where both the pattern and the text are randomly generated. Early studies~\cite{Regnier89,regnier1998complexity} examined the expected number of comparisons in Algorithms \MP and \KMP under memoryless or Markovian source models, employing techniques from analytic combinatorics. Prior attempts based on Markov chains introduced substantial approximations, limiting their accuracy compared to more refined combinatorial methods~\cite{Regnier89,regnier1998complexity}.  Here, we refine and extend this Markov chain-based methodology, providing a more precise foundation for analyzing the expected number of mispredictions (see \cref{sec:mispredictions}). 

Let $\pi$ be a probability measure on $\alphabet$ such that for all $\anyletter\in\alphabet$, $0<\pi(\anyletter) < 1$ (we also use the notation $\pi_\anyletter:=\pi(\anyletter)$ in formulas when convenient). For each $n\geq 0$ and each $\Text\in\alphabet^n$, we define 
$\pi_n(\Text) := \prod_{i=0}^{n-1}\pi(\Text_i)$.
For any $n$, the measure $\pi_n$ is a probability on $\alphabet^n$, where all letters are chosen independently following $\pi$. We obtain the uniform distribution on $\alphabet^n$ if~$\pi$ is the uniform distribution on $\alphabet$, with $\pi(\anyletter)=\frac1{|\alphabet|}$ for all $\anyletter\in\alphabet$.

\vspace{-.3cm}\paragraph*{Encoding the letter comparisons with transducers}
Letter comparisons occur at Line~\ref{line:find-nested-while} only if $i \geq 0$, due to the lazy evaluation of the \texttt{and} operator (when $i < 0$, $\Text[j]$ is not compared to $\Pattern[i]$).  
We can encode these comparisons within the automata of \cref{sec:algorithms} by adding outputs to the transitions, thereby transforming them into transducers. In both $\MX$ and $\KX$, each transition $u \xrightarrow{\anyletter} u\anyletter$ corresponds to matching letters, meaning the test $\Pattern[i] \neq \Text[j]$ evaluates to false. We denote this with the letter $N$ for a \emph{not taken} branch. Conversely, following a failure transition indicates that the test $\Pattern[i] \neq \Text[j]$ is true, which we denote by $T$ for \emph{taken}.  
Transitions from $\bot$ correspond to cases where $i \geq 0$ is false, meaning no letter comparisons occur, as noted above. This construction is illustrated in \cref{fig:transducer failure}.

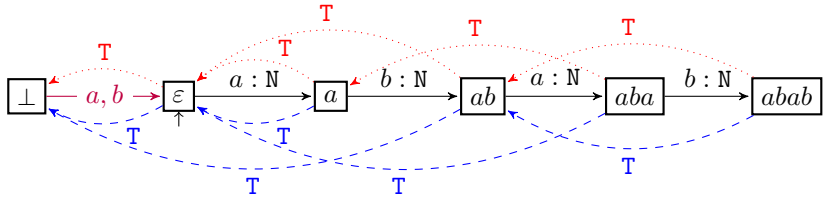
\begin{figure}[t]
    \centering
    %\small
	\begin{tikzpicture}[scale=1]
        \node[draw,thick] (-1) at (-2,0) {\small $\bot$};
        \node[draw,thick] (0) at (0,0) {$\varepsilon$};
        \node[draw,thick] (1) at (2,0) {$a$};
        \node[draw,thick] (2) at (4,0) {$ab$};
        \node[draw,thick] (3) at (6,0) {$aba$};
        \node[draw,thick] (4) at (8,0) {$abab$};
		
		\draw[->] (0,-.4) -- (0);
        \draw (-1) edge[> = stealth',->, shorten >= .1em, purple] node[fill=white]{$a,b$} (0);
		
		\draw (0) edge[> = stealth',->, shorten >= .1em] node[above]{$a:\tt N$} (1);
		\draw (1) edge[> = stealth',->, shorten >= .1em] node[above]{$b:\tt N$} (2);
		\draw (2) edge[> = stealth', ->, shorten >= .1em] node[above,yshift=0.2mm]{$a:\tt N$} (3);
		\draw (3) edge[> = stealth', ->, shorten >= .1em] node[above]{$b:\tt N$} (4);
        
        % MP
        \draw (0) edge[> = stealth',->, shorten >= .1em, bend right, red, dotted] node[above]{\tt T} (-1);
        \draw (1) edge[> = stealth',->, shorten >= .1em, bend right=40, red, dotted] node[above, near start]{\tt T}(0);
        \draw (2) edge[> = stealth',->, shorten >= .1em, bend right=40, red, dotted] node[above]{\tt T} (0);
        \draw (3) edge[> = stealth',->, shorten >= .1em, bend right, red, dotted] node[above]{\tt T} (1);
        \draw (4) edge[> = stealth',->, shorten >= .1em, bend right, red, dotted] node[above]{\tt T} (2);
	    
        % kMP
        \draw (0) edge[> = stealth',->, shorten >= .1em, bend left, blue, dashed] node[below, near start]{\tt T} (-1);
        \draw (1) edge[> = stealth',->, shorten >= .1em, bend left, blue, dashed] node[below, near start]{\tt T} (0);
        \draw (2) edge[> = stealth',->, shorten >= .1em, bend left, blue, dashed] node[below]{\tt T }(-1);
        \draw (3) edge[> = stealth',->, shorten >= .1em, bend left, blue, dashed] node[below]{\tt T} (0);
        \draw (4) edge[> = stealth',->, shorten >= .1em, bend left, blue, dashed] node[below]{\tt T} (2);

	\end{tikzpicture}
    \caption{The automata $\MX$ and $\KX$ transformed into transducers by adding the result of letter comparisons in \textsc{Find} as output of each transition.\label{fig:transducer failure}}
\end{figure}

We keep track of the results of the comparisons $\Pattern[i]\neq \Text[j]$ in $\A_\Pattern$ by simulating the reading of each letter in the transducer associated with $\MX$ and concatenating the outputs. This transforms $\A_\Pattern$ into the transducer $\MTX$ for \MP, by adding an output function $\OutputTransducer_\MTX$ to $\A_\Pattern$ as follows (see \cref{fig:transducers} for an example).
\begin{equation}\label{eq:output}
    \OutputTransducer_\MTX\left(u\xrightarrow{\anyletter}\right) = 
        \begin{cases}
            N & \text{if $u\anyletter\in Q_\Pattern$ or $u\anyletter = \Pattern$},\\
            T & \text{if }u\anyletter\notin Q_\Pattern \text{ and }\fmp(u)= \bot,\\
            T\cdot\OutputTransducer_\MTX\left(\fmp(u)\xrightarrow{\anyletter}\right) & \text{otherwise.}
        \end{cases}
\end{equation}
Instead of $\OutputTransducer_\KTX$ we can use the output $\OutputTransducer_\KTX$ defined as $\OutputTransducer_\MTX$ except that $\fmp$ is changed into $\fkmp$. This yields the transducer $\KTX$ associated with \KMP.

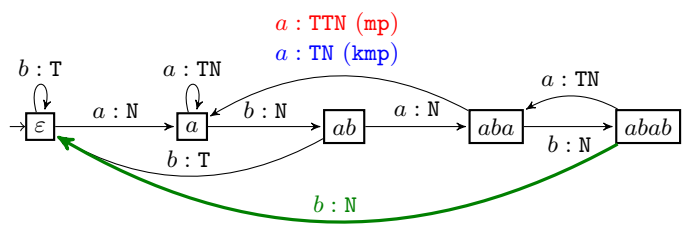
\begin{figure}[t]
	\centering
    \small
	\begin{tikzpicture}[scale=1]
		\node[draw,thick] (0) at (0,0) {$\varepsilon$};
		\node[draw,thick] (1) at (2,0) {$a$};
		\node[draw,thick] (2) at (4,0) {$ab$};
		\node[draw,thick] (3) at (6,0) {$aba$};
		\node[draw,thick] (4) at (8,0) {$abab$};
		
		\draw[->] (-.4,0) -- (0);
		
		\draw (0) edge[> = stealth',->, shorten >= .1em] node[above]{$a:\tt N$} (1);
		\draw (0) edge[> = stealth',loop above] node[above]{$b:\tt T$} (0);
		\draw (1) edge[> = stealth',->, shorten >= .1em] node[above]{$b:\tt N$} (2);
		\draw (1) edge[> = stealth',loop above] node[above]{$a:\tt TN$} (1);
		\draw (2) edge[> = stealth', ->, shorten >= .1em] node[above,yshift=0.2mm]{$a:\tt N$} (3);
		\draw (2) edge[> = stealth', ->, bend left, shorten >= .1em] node[above]{$b:\tt T$} (0);
		\draw (3) edge[> = stealth', ->, bend right, shorten >= .1em] node[above,yshift=0.2mm,align=center]{$\color{red}a:\tt TTN\ (mp)$\\\color{blue}$a:\tt TN\ (kmp)$} (1);
		\draw (3) edge[> = stealth', ->, shorten >= .1em] node[below]{$b:\tt N$} (4);
	    \draw (4) edge[> = stealth', ->, bend right, shorten >= .1em] node[above,yshift=0.2mm]{$a:\tt TN$} (3);
		\draw (4) edge[> = stealth', ->, bend left, shorten >= .1em, very thick, green!50!black] node[above]{$b:\tt N$} (0);
		
	\end{tikzpicture}
    \caption{The transducers $\MTX$ and $\KTX$ for $\Pattern=ababb$. The only difference between them 
    lies in the transition $aba\xrightarrow{a}a$, for which Algorithm \MP uses one more letter comparison.
    \label{fig:transducers}}
\end{figure}

Recall that the output of a path in a transducer is the concatenation of the outputs of its transitions. As the transducers $\MTX$ and $\KTX$ are (input-)deterministic and complete, the output of a word is the output of its unique path that starts at the initial state. From the classical link between $\A_\Pattern$ and Algorithm~\ref{algo:find}~\cite{DBLP:books/ox/CrochemoreR94} we have the following key statement.

\begin{lemma}\label{lm:transducer comparisons}
The sequence of results of the comparisons $\Pattern[i]\neq \Text[j]$ when applying
Algorithm \textsc{Find} to the pattern $\Pattern$ and text $\Text$ is equal to the output of the word $\Text$ in the transducer~$\MTX$ for Algorithm MP, and in the transducer $\KTX$ for \KMP . 
\end{lemma}

\vspace{-.3cm}\paragraph*{State probability and expected number of comparisons}

Since we can use exactly the same techniques, from now on we focus on \KMP for the presentation.
Recall that if we reach the state $u$ after reading the first $j$ letters of $\Text$ in $\A_\Pattern$, and hence in $\KTX$, then at the next iteration of the main \texttt{while} loop, index~$i$ contains the value $|u|$. For $u\in Q_\Pattern$ and $j\in\{0,\ldots,n-1\}$, we are thus interested in the probability $p_\Pattern(j,u)$ that after reading $\pref(\Text,j)$ in $\KTX$ we end in a state $u$. 
Slightly abusing notation, we write $\pi(u)=\pi_{|u|}(u)$. 
For any $u\in Q_X$ let $\border{u}$ denote the longest strict border of $u$, with the convention that $\border{\varepsilon}=\bot$.

\begin{restatable}{lemma}{stationary}\label{lm:stationary probability}
For any $u\in Q_\Pattern$ and any $j\geq m$, $p_\Pattern(j,u)$ does not depend on $j$ and we have
$p_\Pattern(j,u) = p_X(u)$ with $ p_X(u):=\pi(u) - \sum_{\substack{v\in Q_X\\\border{v}=u}}\pi(v)$.
\end{restatable}

From \cref{lm:stationary probability} we can easily estimate the expected number of comparisons for any fixed pattern $\Pattern$, when the length $n$ of $\Text$ tends to infinity. Indeed, except when $j<m$, the probability $p_\Pattern(j,u)$ does not depends on $j$. Moreover, if we are in state $u$, from the length of the outputs of $\KTX$ we can directly compute the expected number of comparisons during the next iteration of the main \texttt{while} loop. See \cref{fig:markov} for a graphical representation.

\begin{figure}[t]
	\centering\small
	\begin{tikzpicture}[scale=1]
		\node[draw,thick] (0) at (0,0) {$4/16$};
		\node[draw,thick] (1) at (2,0) {$6/16$};
		\node[draw,thick] (2) at (4,0) {$3/16$};
		\node[draw,thick] (3) at (6,0) {$2/16$};
		\node[draw,thick] (4) at (8,0) {$1/16$};
		
		\draw (0) edge[> = stealth',->, shorten >= .1em] node[above]{$\frac12:1$} (1);
		\draw (0) edge[> = stealth',loop above] node[above]{$\frac12:1$} (0);
		\draw (1) edge[> = stealth',->, shorten >= .1em] node[above]{$\frac12:1$} (2);
		\draw (1) edge[> = stealth',loop above] node[above]{$\frac12:2$} (1);
		\draw (2) edge[> = stealth', ->, shorten >= .1em] node[above,yshift=0.2mm]{$\frac12:1$} (3);
		\draw (2) edge[> = stealth', ->, bend left=35, shorten >= .1em] node[above]{$\frac12:1$} (0);
		\draw (3) edge[> = stealth', ->, bend right=37, shorten >= .1em] node[above,yshift=0.2mm,align=center]{$\color{red}\frac12:3\ (mp)$\\\color{blue}$\frac12:2\ (kmp)$} (1);
		\draw (3) edge[> = stealth', ->, shorten >= .1em] node[below]{$\frac12:1$} (4);
	    \draw (4) edge[> = stealth', ->, bend right, shorten >= .1em] node[above,yshift=0.2mm]{$\frac12:2$} (3);
		\draw (4) edge[> = stealth', ->, bend left=35, shorten >= .1em, very thick, green!50!black] node[above]{$\frac12: 1$} (0);
		
	\end{tikzpicture}
    \caption{A graphical representation for the computation of the expected number of comparisons for the uniform distribution on $\{a,b\}$: in $\MTX$ and $\KTX$ the state labels $u$ have been changed into their probabilities $p_\Pattern(u)$, the letters into their probabilities $1/2$, and the output into their lengths. For instance, the probability to use the transition $aba\xrightarrow{a}a$ is $\frac2{16}\cdot\frac12=\frac1{16}$ and it yields $3$ comparisons for \MP or 2 for \KMP. Hence its contribution to $C_\Pattern$ in \cref{pro:expected comparisons} is $\frac3{16}$ or $\frac1{8}$.\label{fig:markov}}
\end{figure}

\begin{restatable}{proposition}{procomparisons}\label{pro:expected comparisons}
As $n\rightarrow\infty$, the expected number of letter comparisons performed by Algorithm~\ref{algo:find} with \KMP (or \MP with $\MTX$) is asymptotically equivalent to 
$C_\Pattern\cdot n$, where
\[
C_\Pattern = \sum_{u\in Q_\Pattern} p_X(u)\sum_{a\in\alphabet}\pi(a)\cdot\left|\OutputTransducer_\KTX\left(u\xrightarrow{a}\right)\right|,
\text{ and }1\leq C_\Pattern\leq 2.
\]
\end{restatable}

Observe that \cref{lm:stationary probability} can also be derived by transforming $\KTX$ into a Markov chain and computing its stationary distribution~\cite{LePeWe08}. However, \cref{lm:stationary probability} provides a more direct and simpler formula, which appears to have gone unnoticed in the literature. Markov chains will prove very useful in \cref{sec:mispredictions comparisons}.

\section{Expected number of mispredictions}\label{sec:mispredictions}

We now turn to our main objective: a theoretical analysis of the number of mispredictions for a fixed pattern $\Pattern$ and a random text $\Text$.  
The analysis depends on the type of predictor used (see \cref{sec:introduction}). In what follows, all results are stated for local 2-bit saturated counters, such as the one in \cref{fig:predictor}.  
Let $\xi$ denote its transition function extended to binary words. For example, $\xi(\sN, NNNTT) = \sT$.  
Additionally, let $\mu(\lambda, s)$ denote the number of mispredictions encountered when following the path in the predictor starting from state $\lambda \in \{\sNN, \sN, \sT, \sTT\}$ and labeled by $s \in \{N, T\}^*$.  
For instance, $\mu(\sN, NNNTT) = 2$.

As previously noted, Algorithm~\ref{algo:find} contains four branches in total: one at Line~\ref{line:find-main-while}, two at Line~\ref{line:find-nested-while}, and one at Line~\ref{line:find-counter-if}. Each of these branches is assigned a local predictor, and all have the potential to generate mispredictions.  
The mispredictions generated by the main \texttt{while} loop (i.e. Line~\ref{line:find-main-while}) are easily analyzed. Indeed, the test holds true for $n$ times and then becomes false. Hence, the sequence of taken/not taken outcomes for this branch is~$T^nN$. 
Therefore, starting from any state of the 2-bit saturated predictor, at most three mispredictions can occur. It is asymptotically negligible, as we will demonstrate that the other branches produce a linear number of mispredictions on average.

\subsection{Mispredictions of the counter update}\label{sec:mispred counter}

We analyze the expected number of mispredictions induced by the counter update at Line~\ref{line:find-counter-if}. The sequence $s$ of taken/not-taken outcomes for this \texttt{if} statement is defined by $s_j=T$ if and only if $\pref(\Text,j)$ ends with the pattern $\Pattern$, for all $j\in\{0,\ldots, n-1\}$. This is easy to analyze, especially when the pattern $\Pattern$ is not the repetition of a single letter.
Proposition~\ref{pro:mispred counter typical} establishes that, on average, there is approximately one misprediction for each occurrence of the pattern in the text.

\begin{restatable}{proposition}{countertypical}\label{pro:mispred counter typical}

If $\Pattern$ contains at least two distinct letters, then the expected number of mispredictions caused by the counter update is asymptotically equivalent to $\pi(\Pattern)\cdot n$.
\end{restatable}
\begin{proof}[Proof (sketch)]
Since $\Pattern$ contains at least two distinct letters, it cannot be a suffix of both $\pref(\Text,j)$ and $\pref(\Text,j+1)$. 
Hence, the sequence $s$ is of the form $(N^+T)^*N^*$. 
This means that every step to the right in the local predictor (for every $T$ in sequence~$s$), which corresponds to a match, is followed by a step to the left, except possibly for the last step.
Thus, if the local predictor reaches state $\sNN$, it remains in $\sN$ or $\sNN$ forever. 
Having three consecutive positions in $\Text$ without an occurrence of $\Pattern$ is sufficient to reach state $\sNN$. This happens in fewer than $\O(\log n)$ iterations with high probability, and at this point there is exactly one misprediction each time the pattern is found. This concludes the proof, as the expected number of occurrences of $\Pattern$ in $\Text$ is asymptotically equivalent to $\pi(\Pattern)\cdot n$. 
\end{proof}

The analysis of the case $\Pattern=\alpha^m$, where $\Pattern$ consists of a repeated single letter, is more intricate. We first present the proof sketch for $\Pattern=\alpha\alpha$, which captures all the essential ideas.
Let $A'=\alphabet\setminus\{\alpha\}$
and write $\Text= \beta_1 \beta_2 \dots \beta_\ell \alpha^{x}$, where $\beta_i=\alpha^{k_i} \overline{\alpha}$ with $k_i \geq 0$ and $\overline{\alpha} \in A'$. 
Depending on the value of $k_i$, one can compute the sequence of taken/not taken outcomes induced by a factor $\alpha^{k_i}\overline\alpha$, which is either preceded by a letter $\overline\alpha$ or nothing: $\overline\alpha$ yields $N$, $\alpha\overline\alpha$ yields $NN$, $\alpha^2\overline\alpha$ yields $NTN$, and so on. Thus, more generally, $\overline\alpha$ yields $N$ and $\alpha^{k_i}\overline\alpha$ yields $NT^{k_i-1}N$ for $k_i\geq 1$. 
We then examine the state of the predictor and the number of mispredictions produced after each factor $\beta_i$ is read. For instance, if just before reading $\beta_i=\alpha^3\overline\alpha$ the predictor state is $\sN$, then the associated sequence $NTTN$ produces three mispredictions and the predictor ends in the same state $\sN$, which can be seen on the path $\sN\xrightarrow[\phantom{misp.}]{N} \sNN\xrightarrow[misp.]{T}\sN\xrightarrow[misp.]{T}\sT\xrightarrow[misp.]{N}\sN$. Since $\sTT$ cannot be reached except at the very beginning or at the very end, it has a negligible contribution to the expectation, and we can list all the relevant possibilities as follows:
\[
    \begin{array}{c|l|l|l|l|l|l|l|l|l|l}
       & \multicolumn{2}{c|}{k=0} &  \multicolumn{2}{c|}{k=1} & \multicolumn{2}{c|}{k=2}&  \multicolumn{2}{c|}{k=3}&  \multicolumn{2}{c}{k\geq 4}\\
        & \multicolumn{2}{c|}{N} & \multicolumn{2}{c|}{NN} & \multicolumn{2}{c|}{NTN} & \multicolumn{2}{c|}{NTTN} & \multicolumn{2}{c}{NT^{k-1}N}\\ 
       \hline
       \sNN & 
        \rightarrow\sNN & 0 \text{ misp.}&
        \rightarrow \sNN  & 0 \text{ misp.}&
        \rightarrow \sNN  & 1 \text{ misp.}&
        \rightarrow \sN  & 3 \text{ misp.}&
        \rightarrow \sT & 3 \text{ misp.}\\
       \sN &
        \rightarrow \sNN & 0 \text{ misp.}&
        \rightarrow \sNN  & 0 \text{ misp.}&
        \rightarrow \sNN  & 1 \text{ misp.}&
        \rightarrow \sN  & 3 \text{ misp.}&
        \rightarrow \sT & 3 \text{ misp.}\\
       \sT &
        \rightarrow \sN & 1 \text{ misp.}&
        \rightarrow \sNN  & 1 \text{ misp.}&
        \rightarrow \sN  & 3 \text{ misp.}&
        \rightarrow \sT  & 3 \text{ misp.}&
        \rightarrow \sT & 3 \text{ misp.}
    \end{array}
\]

In the table above, the states $\sNN$ and $\sN$ produce identical outcomes and can therefore be merged into a single state, denoted as $\tilde{\sN}$, for the analysis. The resulting transitions form a graph with two vertices, which is then converted into a Markov chain by incorporating the transition probabilities $\alpha^k \overline{\alpha}$, as illustrated in \cref{fig:aa}.

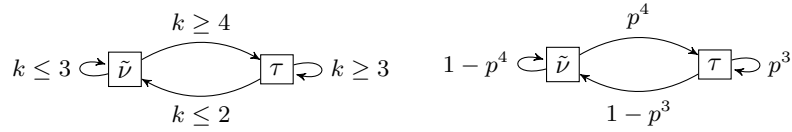
\begin{figure}[t]
    \centering
    \begin{tikzpicture}
        \node[draw] (N) at (0,0) {$\tilde\sN$};
        \node[draw] (T) at (2,0) {$\sT$};

        \draw (N) edge[> = stealth',loop left] node {\small $k\leq 3$} (N);
        \draw (T) edge[> = stealth',->,bend left] node[below] {\small $k\leq 2$} (N);
        \draw (T) edge[> = stealth',loop right] node {\small $k\geq 3$} (T);
        \draw (N) edge[> = stealth',->,bend left] node[above] {\small $k\geq 4$} (T);
    \end{tikzpicture}
    \quad
    \begin{tikzpicture}
        \node[draw] (N) at (0,0) {$\tilde\sN$};
        \node[draw] (T) at (2,0) {$\sT$};

        \draw (N) edge[> = stealth',loop left] node {\small $1-p^4$} (N);
        \draw (T) edge[> = stealth',->,bend left] node[below] {\small $1-p^3$} (N);
        \draw (T) edge[> = stealth',loop right] node {\small $p^3$} (T);
        \draw (N) edge[> = stealth',->,bend left] node[above] {\small $p^4$} (T);
    \end{tikzpicture}
    \caption{\label{fig:aa}%
        On the left, the transition system determined by the factor $\alpha^k \overline{\alpha}$; on the right, the corresponding Markov chain. For clarity, we denote $p := \pi(\alpha)$.
    }
\end{figure}

The stationary distribution $\pi_0$ of this Markov chain is straightforward to compute, yielding $\pi_0(\tilde\sN)= \frac{1-p^3}{1-p^3+p^4}$ and $\pi_0(\sT)=\frac{p^4}{1-p^3+p^4}$, where $p:=\pi(\alpha)$. From each state, the expected number of mispredictions can be computed using the transition table.  
For instance, starting from $\tilde\sN$, a misprediction occurs when $k=2$ with probability $(1-p)p^2$, and three mispredictions occur when $k\geq 3$ with probability $p^3$.
Therefore, the expected number of mispredictions when reading the next factor $\alpha^k\overline\alpha$ from $\tilde\sN$ is given by $(1-p)p^2 + 3p^3$. 
Finally, with high probability there are around $(1-p)n$ factors of the form $\anyletter^*\overline\anyletter$ in the decomposition of $\Text$, which corresponds to roughly the same number of steps in the Markov chain. The general statement for $\Pattern=\anyletter^m$ is as follows.

\begin{restatable}{proposition}{misspredam}\label{pro:mispred counter a^m}
If $\Pattern=\alpha^m$,  the expected number of mispredictions caused by the counter update is asymptotically $\kappa_m(\pi(\alpha))\cdot n$, with
$\kappa_m(p) = p^m(1 - p)(1 + p)^2$ for $m \geq 3$,
and 
\[
\kappa_1(p) = \frac{p(1-p)}{1-2p(1-p)},\quad\text{ and }\quad
\kappa_2(p) = \frac{p^2(1-p)\left(1+2p+p^2-p^3\right)}{1-p^3+p^4}.
\]
\end{restatable}

\subsection{Expected number of mispredictions during letter comparisons}\label{sec:mispredictions comparisons}

In this section, we analyze the expected number of mispredictions caused by letter comparisons in \KMP (similar results can be derived for \MP).

According to \cref{lm:transducer comparisons}, the outcome of letter comparisons in \KMP is encoded by the transducer $\KTX$. More precisely, following a transition $u\xrightarrow{\anyletter:s}v$ in this transducer simulates a single  iteration of the main loop of Algorithm~\ref{algo:find}, starting with $i=|u|$ and processing the letter $\anyletter:=\Text[j]$. At the end of this iteration, $i=|v|$, and $s\in \{N,T\}^*$ is the sequence of taken/not-taken outcomes for the test $\Pattern[i]\neq \Text[j]$. 

The mispredictions occurring during this single iteration of the main loop depend on the predictor's initial state $\lambda$ and the sequence~$s$ which is computed using $\KTX$. The number of mispredictions~$\mu(\lambda, s)$ is retrieved by following the path starting from state $\lambda$ and labeled by~$s$ in the predictor, corresponding to the transition $\xi(\lambda, s)$. 
This is formalized by using a coupling of $\KTX$ with the predictor in \cref{fig:predictor}, forming a product transducer $\PKX$, defined as follows (see \cref{fig:PX} for an example):
\begin{itemize}
    \item the set of states is $Q_\Pattern\times \{\sNN,\sN,\sT,\sTT\}$, %where $Q_\Pattern$ is the set of states of $\KTX$,
    \item there is a transition $(u,\lambda) \xrightarrow{\anyletter:\mu(\lambda,s)}(\delta_\Pattern(u),\xi(\lambda,s))$ for every state $(u,\lambda)$ and every letter~$\anyletter$, where $s$
     is the output of the transition $u\xrightarrow{\anyletter:s}\delta_\Pattern(u)$ in $\KTX$.
\end{itemize}
By construction, at the beginning of an iteration of the main loop in Algorithm~\ref{algo:find}, if~$i=|u|$, $\lambda$ is the initial state of the 2-bit saturated predictor, and $\alpha=\Text[j]$,  
then, during the next iteration, $\mu(\lambda,s)$ mispredictions occur, and the predictor terminates in state $\xi(\lambda,s)$,
where $u\xrightarrow{\anyletter:s}\delta_\Pattern(u)$ in $\KTX$. 
This leads to the following statement.

\begin{lemma}\label{lm:P_X}
The number of mispredictions caused by letter comparisons in \KMP, when applied to the text $\Text$ and the pattern $\Pattern$, is given by the sum of the outputs along the path that starts at $(\varepsilon,\lambda_0)$ and is labeled by $\Text$ in $\PKX$, where $\lambda_0\in\{\sNN,\sN,\sT,\sTT\}$ is the initial state of the local predictor associated with the letter comparison.
\end{lemma}

\begin{figure}[h]
    \centering
    \small
    \begin{tikzpicture}[xscale=1.2,yscale=0.9]      
        \node[draw,thick] (0N) at (0,3) {$\varepsilon,\sNN$};
        \node[draw,thick] (0n) at (0,2) {$\varepsilon,\sN$};
        \node[draw,thick] (0t) at (0,1) {$\varepsilon,\sT$};
        \node[draw,thick] (0T) at (0,0) {$\varepsilon,\sTT$};
        
        \node[draw,thick] (1N) at (2,3) {$a,\sNN$};
        \node[draw,thick] (1n) at (2,1.5) {$a,\sN$};
        \node[draw,thick] (1t) at (2,0) {$a,\sT$};
        
        \node[draw,thick] (2N) at (4,3) {$ab,\sNN$};
        \node[draw,thick] (2n) at (4,0) {$ab,\sN$};
        
        \node[draw,thick] (3N) at (6,3) {$aba,\sNN$};
        \node[draw,thick] (4N) at (8,3) {$abab,\sNN$};

        \draw (0N) edge[->,> = stealth'] node[left] {$b:1$} (0n);
        \draw (0n) edge[->,> = stealth'] node[left] {$b:1$} (0t);
        \draw (0t) edge[->,> = stealth'] node[left] {$b:0$} (0T);
        \draw (0T) edge[> = stealth',loop left] node[left] {$b:0$} (0T);

        \draw (0T) edge[> = stealth',->, shorten >= .1em] node[below] {$a:1$} (1t);
        \draw (0t) edge[> = stealth',->, shorten >= .1em] node[above] {$a:1$} (1n);
        \draw (0n) edge[> = stealth',->, shorten >= .1em] node[above] {$a:0$} (1N);
        \draw (0N) edge[> = stealth',->, shorten >= .1em] node[above] {$a:0$} (1N);

        \draw (1N) edge[> = stealth',loop above] node[above]{$a:1$} (1N);
        \draw (1n) edge[->,> = stealth',loop, in=-20,out=-60,min distance=5mm] node[right]{$a:2$} (1n);
        \draw (1t) edge[> = stealth',loop below] node[below]{$a:1$} (1t);

        \draw (1N) edge[> = stealth',->, shorten >= .1em] node [above, near end]{$b:0$} (2N);
        \draw (1n) edge[> = stealth',->, shorten >= .1em] node[below]{$b:0$} (2N);
        \draw (1t) edge[> = stealth',->, shorten >= .1em] node[below]{$b:1$} (2n);

        \draw (2N) edge[> = stealth',->, shorten >= .1em, bend left=12] node[above]{$b:1$} (0n);
        \draw (2n) edge[> = stealth',->, shorten >= .1em, bend right=12] node[below]{$b:1$} (0t);

        \draw (2N) edge[> = stealth',->, shorten >= .1em] node[below]{$a:0$} (3N);
        \draw (2n) edge[> = stealth',->, shorten >= .1em] node[left]{$a:0$} (3N);

        \draw[blue] (3N) edge[> = stealth',->, shorten >= .1em, bend right=30] node[above]{$a:1$ (kmp)} (1N);
        \draw[red] (3N) edge[> = stealth',->, shorten >= .1em, bend left=20] node[above, rotate=18]{$a:3$ (mp)} (1n);

        \draw ([yshift=.8mm]4N.west) edge[> = stealth',->, shorten >= .1em, bend right=12] node[above]{$a:1$} ([yshift=.8mm]3N.east);
        \draw ([yshift=-.8mm]3N.east) edge[> = stealth',->, shorten >= .1em, bend right=12] node[below]{$b:0$} ([yshift=-.8mm]4N.west);
        
        \draw (4N) edge[> = stealth',->, shorten >= .1em, bend right=40] node[above]{$b:0$} (0N);

    \end{tikzpicture}
    \caption{The strongly connected terminal component of $\PKX$ in black and blue, for $\Pattern=ababb$. In black and red, the variant for~$\PX$.\label{fig:PX}}
\end{figure}

We can then proceed as in Proposition~\ref{pro:mispred counter a^m}: the transducer $\PKX$ is converted into a Markov chain by assigning a weight of $\pi(\alpha)$ to the transitions labeled by a letter $\anyletter$. From this, we compute the stationary distribution $\pi_0$ over the set of states, allowing us to determine the asymptotic expected number of mispredictions per letter of $\Text$.
This quantity, $L_\Pattern$, satisfies
\begin{equation}\label{eq:mispred comparisons}
L_\Pattern = \sum_{u\in Q_\Pattern}\sum_{\lambda\in\{\sNN,\sN,\sT,\sTT\}}
\pi_0(u,\lambda) \times \sum_{\anyletter\in\alphabet} \pi(\anyletter)\cdot \OutputTransducer_{\PKX}((u,\lambda)\xrightarrow{\anyletter}).
\end{equation}

Observe that when processing a long sequence of letters different from $\Pattern[0]$, the letter comparisons produce a sequence of $T$'s, causing the 2-bit saturated predictor to settle in state $\sTT$ while $i=0$ in the algorithm. Consequently, the state $(\varepsilon,\sTT)$ is reachable from every other state. Hence, the Markov chain has a unique terminal strongly connected component (i.e. there are no transitions from any vertex in this strongly connected component to any vertex outside of it), which includes $(\varepsilon,\sTT)$ along with a self-loop at this state. Thus, our analysis focuses on this component, allowing us to apply classical results on primitive Markov chains~\cite{LePeWe08}, ultimately leading to Equation~(\ref{eq:mispred comparisons}). 
Notably, this result is independent of the predictor's initial state.
The computation of $L_\Pattern$ can be easily carried out using computer algebra, since computing the stationary probability reduces to inverting a matrix.

\begin{proposition}\label{pro:mispred comparisons}
The expected number of mispredictions caused by letter comparisons in \KMP on a random text of length $n$ and a pattern $\Pattern$, is asymptotically equivalent to $L_\Pattern\cdot n$.
\end{proposition}

\subsection[Expected number of mispredictions of the first test of the inner while]{Expected number of mispredictions of the test $i\geq 0$}\label{sec:mispredictions i>=0} 

We conclude the analysis by examining the mispredictions caused by the test $i\geq 0$ at Line~\ref{line:find-nested-while} of Algorithm~\ref{algo:find}. To this end, we use the previously constructed transducer $\KTX$ (or equivalently $\MTX$, as the approach remains the same) to capture the behavior of this test through a straightforward transformation of the outputs.
Recall that a transition $u\xrightarrow{\anyletter:s}v$ in $\KTX$, with $s\in\{N,T\}^*$ indicates that when reading the letter $\anyletter$, the inner {\tt while} loop performs $|s|$ character comparisons, with the result encoded by the symbols of $s$. Due to the loop structure, $s$ always takes one of two forms:
\begin{itemize}
    \item $T^*N$ and the loop terminates because $\Pattern[i]=\Text[j]$ eventually, or
    \item $T^+$ and the loop terminates because $i=-1$ eventually.
\end{itemize}
In the first case, the condition $i\geq 0$ holds for $|s|$ iterations. In the second case, the condition holds for $|s|$ iterations before failing once. 
Thus, we define the transducer $\WKTX$ identically to $\KTX$, except for its output function:
\begin{equation}\label{eq:i>=0}
\OutputTransducer_{\WKTX}\left(u\xrightarrow{\anyletter}\right) =
    \begin{cases}
        T^{|s|}& \text{if }s=\OutputTransducer_{\KTX}\left(u\xrightarrow{\anyletter}\right) \in T^*N,\\
        T^{|s|}N& \text{if }s=\OutputTransducer_{\KTX}\left(u\xrightarrow{\anyletter}\right) \in T^+.
    \end{cases}
\end{equation}
The same transformation can be applied to $\MTX$ for \MP. At this stage, we could directly reuse the framework from Section~\ref{sec:mispredictions comparisons} to compute the asymptotic expected number of mispredictions for any given pattern $\Pattern$. However, a shortcut allows for a simpler formulation while offering deeper insight into the mispredictions caused by the test $i\geq 0$.

Since each output is either $T^kN$ for some $k \geq 1$ or $T^k$, the local predictor state generally moves toward $\sTT$, except in the case of $TN$.  
In this latter case, the predictor either remains in the same state or transitions from $\sTT$ to $\sT$.  
Moreover, from any state $s$ of $\A_\Pattern$, there always exists a letter $\anyletter$ such that $s \xrightarrow{\alpha:T}$ in $\WMTX$ or $\WKTX$ (for instance, the transition that goes to the right or when the pattern is found).  
As a result, with high probability, the predictor reaches the state $\sTT$ in at most $\mathcal{O}(\log n)$ iterations of the main loop of Algorithm~\ref{algo:find}.  
Once in $\sTT$, the predictor remains confined to the states $\sT$ and $\sTT$ indefinitely.  
Thus, with high probability, except for a small number of initial steps, the predictor consistently predicts that the branch is taken.  
At this point, a misprediction occurs if and only if the output belongs to $T^*N$, which happens precisely when a non-accepting transition in $\WKTX$ leads to the state $\varepsilon$.  
Since $\WKTX$ and $\WMTX$ differ only in their output functions, this result holds for both \MP and \KMP, allowing us to work directly with $\A_\Pattern$.  
Applying \cref{lm:stationary probability}, we obtain the following statement.

\begin{restatable}{proposition}{inonnegative}\label{pro:mispred i>=0}
When Algorithms \MP or \KMP are applied to a random text $\Text$ of length~$n$ with a given pattern $\Pattern$, the expected number of mispredictions caused by the test $i \geq 0$ is equal to the expected number of times a transition ending in $\varepsilon$ is taken along the path labeled by $\Text$ in $\A_\Pattern$, up to an error term of $\O(\log n)$. 
As a result, the expected number of such mispredictions is asymptotically equivalent to $G_\Pattern\cdot n$, where
$G_\Pattern = \sum_{u\in Q_\Pattern}p_X(u) \sum_{\substack{u\xrightarrow{\anyletter}\varepsilon\\u\anyletter\neq X}} \pi(\anyletter)$.
\end{restatable}

\section{Results for small patterns, discussion and perspectives}\label{sec:discussion}

\begin{table}[htbp!]
    \centering
    \renewcommand{\arraystretch}{1.75}
    \scalebox{0.96}{%
        \begin{tabular}{ccccc}
        %\toprule
        \Pattern  & \texttt{i = m} & \texttt{i >= 0} & Algo. & \texttt{X[i] != T[j]}  \\
        \midrule
        % aa
        \multirow{2}{*}{aa}  & \multirow{2}{*}{${\displaystyle\kappa_2(p)  }$} & \multirow{2}{*}{${\displaystyle 1 - p}$} & MP &  ${\displaystyle \frac{p (1-p)(1+2p)}{1-p^2+p^{3}}}$  \\
        &  &  & KMP & ${\displaystyle \frac{p (1 - p)}{1-2p + 2p^{2}}}$ \\
        \midrule
        % ab
        ab & ${\displaystyle p (1 - p)}$ & ${\displaystyle (1-p)^2}$ & both &  ${\displaystyle \frac{p (3-7p+7p^2-2 p^{3})}{1-p+2p^2-p^{3}}}$ \\
        \midrule
        % aaa
        \multirow{2}{*}{aaa} & \multirow{2}{*}{${\displaystyle \kappa_3(p)%p^{3} (1 - p) \left(p + 1\right)^{2}
        }$} & \multirow{2}{*}{${\displaystyle 1 - p}$} & MP &  ${\displaystyle p (1-p)(1+p)^2}$ \\
        & & & KMP & ${\displaystyle \frac{p (1 - p)}{1-2p+2 p^{2}}}$ \\
        \midrule
        % aab
        \multirow{2}{*}{aab}  &  \multirow{2}{*}{$p^{2} \left(1 - p\right)$} & \multirow{2}{*}{$(1 - p)^2(1+p)$} & MP & ${\displaystyle p (1+2p-p^2-8p^3+6p^4+5p^5-5p^6+p^7)}$ \\
        & & & KMP & ${\displaystyle \frac{p (1-2p^2 -p^3+5p^4-3p^5+p^6)}{1-2p+3p^2-2p^3 +p^{4}}}$\\
        \midrule
        % aba
        \multirow{2}{*}{aba} &  \multirow{2}{*}{${\displaystyle p^{2} \left(1 - p\right)}$} & \multirow{2}{*}{${\displaystyle (1-p)^{2}}$} & MP & ${\displaystyle \frac{p (3-7p+8p^2-4p^3+p^{4})}{1-p+p^{2}}}$ \\
        & & & KMP & ${\displaystyle \frac{p (3-7p+7p^2-2 p^{3})}{1-p+2p^2-p^{3}}}$\\
        \midrule
        % abb
        abb &  ${\displaystyle p (1-p)^{2}}$ & ${\displaystyle (1-p)^{3}}$ & both & ${\displaystyle p (4-13p+21p^2-16p^3+6p^4-p^5)}$ \\
        \bottomrule
        \end{tabular}
    }% end scalebox
    \vspace{.25cm}
    \caption{\label{fig:result k=2}%
        We analyze the asymptotic expected number of mispredictions per symbol in the text for each branch of Algorithm~\ref{algo:find}, considering $\Sigma = \{a, b\}$ and all normalized patterns of length~2 and 3. To improve readability, we introduce $p := \pi(a) = 1 - \pi(b)$. Notably, for the patterns $\Pattern = ab$ and $\Pattern = abb$, the failure functions of Algorithms \MP and \KMP are identical, making both variants of Algorithm \textsc{Find} behave the same in these cases. The functions $\kappa_2$ and $\kappa_3$ are defined in \cref{pro:mispred counter a^m}.
    }
\end{table}

We conducted a comprehensive study of local branch prediction for MP and KMP and provide %the code\footnote{\texttt{Python} notebook (using \texttt{sympy}), available at \url{https://github.com/vialette/branch-prediction/}} 
algorithms that allows to quantify mispredictions for any alphabet size, any given pattern and any memoryless source for the input text (as for the examples given in \cref{fig:result k=2}).

Notably, the expressions for the number of mispredicted letter comparisons become increasingly complex as the pattern length grows and as the alphabet size increases. For instance, for the pattern $X = abab$, with $\pi_a := \pi(a)$ and $\pi_b := \pi(b)$, we obtain:
\[
L_{abab} = \frac{\pi_a (- \pi_a^{3} \pi_b + 2 \pi_a^{2} \pi_b^{3} + 4 \pi_a^{2} \pi_b^{2} + 3 \pi_a^{2} \pi_b + \pi_a^{2} - 5 \pi_a \pi_b^{2} - 4 \pi_a \pi_b - 2 \pi_a + 2 \pi_b + 1)}{(1-\pi_a)(\pi_a^{2} \pi_b^{2} + \pi_a^{2} \pi_b - \pi_a \pi_b - \pi_a + 1)}.\]

The results given in \cref{fig:result uniform} illustrate this for the uniform distribution, for small patterns and alphabets. In particular, the branch $i\geq 0$, which is poorly predicted by its local predictor, exhibits a very high number of mispredictions when $|\alphabet| = 4$, while the branch that comes from letter comparisons, $X[i] \neq T[j]$, experiences fewer mispredictions. This trend becomes more pronounced as the size of the alphabet increases: for
$\Pattern = abb$ and $|\alphabet| = 26$, the misprediction rate for the test $i\geq 0$ reaches~0.96, whereas for  $X[i] \neq T[j]$, it drops to 0.041.

Our work presents an initial theoretical exploration of pattern matching algorithms within computational models enhanced by local branch prediction. However, modern processors often employ hybrid prediction mechanisms that integrate both local and global predictors, with global predictors capturing correlations between branch outcomes across different execution contexts. A key direction for further research is to develop a theoretical model that incorporates both predictors, allowing for more precise measurement in real-world scenarios. Another important line of research is to account for enhanced probabilistic distributions for texts, as real-word texts are often badly modeled by memoryless sources. For instance, Markovian sources should be manageable within our model and could provide a more accurate framework for the analysis.

\begin{table}[htbp]
    \centering
    \renewcommand{\arraystretch}{1.1}
    \scalebox{.96}{%
%    \begin{tabular}{c|ccccl|ccccl}
    \begin{tabular}{ccccclccccl}
    & \multicolumn{5}{c}{$|\alphabet|=2$} &\multicolumn{5}{c}{$|\alphabet|=4$} \\
    \cmidrule(lr){2-6}
    \cmidrule(lr){7-11}
    \Pattern  & {\tt i=m} & {\tt i>=0} & algo & {\tt X[i]!=T[j]} & Total& {\tt i=m} & {\tt i>=0} & algo & {\tt X[i]!=T[j]} & Total\\
    \cmidrule(lr){2-6}
    \cmidrule(lr){7-11}
    % aa
    \multirow{2}{*}{aa} & \multirow{2}{*}{0.283} & \multirow{2}{*}{0.5} & \MP & 0.571 & 1.353 & \multirow{2}{*}{0.073} & \multirow{2}{*}{0.75} & \MP & 0.295 & 1.117 \\
      & & & \KMP &0.5 & 1.283 & & & \KMP &0.3 & 1.123\\
    \cmidrule(lr){2-6}
    \cmidrule(lr){7-11}
    % ab
    ab & 0.25 & 0.25 & both & 0.571 & 1.321 & 0.062 & 0.688 & both & 0.375 & 1.186\\ 
    \cmidrule(lr){2-6}
    \cmidrule(lr){7-11}
    % aaa
    \multirow{2}{*}{aaa} & \multirow{2}{*}{0.14} & \multirow{2}{*}{0.5} & \MP & 0.563 & 1.202 & \multirow{2}{*}{0.018} & \multirow{2}{*}{0.75} & \MP & 0.293 & 1.06 \\
      & & & \KMP &0.5 & 1.14 & & & kmp &0.3 & 1.068\\
    \cmidrule(lr){2-6}
    \cmidrule(lr){7-11}
    % aab
    \multirow{2}{*}{aab} & \multirow{2}{*}{0.125} & \multirow{2}{*}{0.375} & \MP & 0.605 & 1.23 & \multirow{2}{*}{0.015} & \multirow{2}{*}{0.734} & \MP & 0.322 & 1.086 \\
      & & & \KMP &0.542 & 1.166 & & & \KMP &0.322 & 1.086\\
    \cmidrule(lr){2-6}
    \cmidrule(lr){7-11}
    % aba
    \multirow{2}{*}{aba} & \multirow{2}{*}{0.125} & \multirow{2}{*}{0.25} & \MP & 0.708 & 1.083 & \multirow{2}{*}{0.015} & \multirow{2}{*}{0.688} & \MP & 0.367 & 1.068 \\
      & & & \KMP &0.571 & 0.946 & & & \KMP &0.375 & 1.076\\
    \cmidrule(lr){2-6}
    \cmidrule(lr){7-11}
    % abb
    abb & 0.125 & 0.125 & both & 0.547 & 0.921 & 0.015 & 0.672 & both & 0.397 & 1.098\\ 
    \cmidrule(lr){2-6}
    \cmidrule(lr){7-11}
    \end{tabular}
    }% end scalebox
    \vspace{.25cm}
    \caption{\label{fig:result uniform}%
        Asymptotic expected number of mispredictions per input symbol in a random text~$\Text$, using Algorithm~\ref{algo:find}, assuming a uniform distribution over alphabets of size 2 and 4.  
    }
\end{table}

\newpage

%%%%%%%%%%%%%%% FIN DU DOCUMENT %%%%%%%%%%%

\newpage

\bibliography{biblio}

\begin{thebibliography}{10}

\bibitem{AuNiPi2016}
Nicolas Auger, Cyril Nicaud, and Carine Pivoteau.
\newblock {Good predictions are worth a few comparisons}.
\newblock In {\em {STACS 2016}}, volume~47, pages 12:1--12:14, Orl{\'e}ans,
  France, February 2016.
\newblock URL: \url{https://hal.science/hal-01212840}, \href
  {https://doi.org/10.4230/LIPIcs.STACS.2016.12}
  {\path{doi:10.4230/LIPIcs.STACS.2016.12}}.

\bibitem{BrMo05}
Gerth~Stølting Brodal and Gabriel Moruz.
\newblock Tradeoffs {Between} {Branch} {Mispredictions} and {Comparisons} for
  {Sorting} {Algorithms}.
\newblock In {\em Algorithms and {Data} {Structures}}, volume 3608, pages
  385--395. Springer Berlin Heidelberg, Berlin, Heidelberg, 2005.

\bibitem{DBLP:books/daglib/0020103}
Maxime Crochemore, Christophe Hancart, and Thierry Lecroq.
\newblock {\em Algorithms on strings}.
\newblock Cambridge University Press, 2007.

\bibitem{DBLP:books/ox/CrochemoreR94}
Maxime Crochemore and Wojciech Rytter.
\newblock {\em Text Algorithms}.
\newblock Oxford University Press, 1994.
\newblock URL: \url{http://www-igm.univ-mlv.fr/\%7Emac/REC/B1.html}.

\bibitem{DBLP:books/cu/Gusfield1997}
Dan Gusfield.
\newblock {\em Algorithms on Strings, Trees, and Sequences - Computer Science
  and Computational Biology}.
\newblock Cambridge University Press, 1997.
\newblock URL: \url{https://doi.org/10.1017/cbo9780511574931}, \href
  {https://doi.org/10.1017/CBO9780511574931}
  {\path{doi:10.1017/CBO9780511574931}}.

\bibitem{HePa17}
John~L. Hennessy and David~A. Patterson.
\newblock {\em Computer Architecture, Sixth Edition: A Quantitative Approach}.
\newblock Morgan Kaufmann Publishers Inc., 6th edition, 2017.

\bibitem{knuth1977fast}
Donald~E Knuth, James~H Morris, Jr, and Vaughan~R Pratt.
\newblock Fast pattern matching in strings.
\newblock {\em SIAM journal on computing}, 6(2):323--350, 1977.

\bibitem{LePeWe08}
David~A. Levin, Yuval Peres, and Elizabeth~L. Wilmer.
\newblock {\em {Markov Chains and Mixing Times}}.
\newblock American Mathematical Society, 2008.
\newblock URL: \url{http://pages.uoregon.edu/dlevin/MARKOV/markovmixing.pdf}.

\bibitem{MaNeWi15}
Conrado Mart{\'{\i}}nez, Markus~E. Nebel, and Sebastian Wild.
\newblock Analysis of branch misses in quicksort.
\newblock In {\em Proceedings of the Twelfth Workshop on Analytic Algorithmics
  and Combinatorics, {ANALCO} 2015, San Diego, CA, USA, January 4, 2015}, pages
  114--128, 2015.
\newblock \href {https://doi.org/10.1137/1.9781611973761.11}
  {\path{doi:10.1137/1.9781611973761.11}}.

\bibitem{Mittal2018}
Sparsh Mittal.
\newblock A survey of techniques for dynamic branch prediction.
\newblock {\em Concurrency and Computation: Practice and Experience}, 31, 2018.
\newblock URL: \url{https://api.semanticscholar.org/CorpusID:4572006}.

\bibitem{MP}
James~H Morris, Jr and Vaughan~R Pratt.
\newblock {A Linear Pattern-Matching Algorithm}.
\newblock Technical report, University of California, Berkeley, CA, 01 1970.

\bibitem{Regnier89}
Mireille R{\'{e}}gnier.
\newblock Knuth-morris-pratt algorithm: An analysis.
\newblock In Antoni Kreczmar and Grazyna Mirkowska, editors, {\em Mathematical
  Foundations of Computer Science 1989, MFCS'89, Porabka-Kozubnik, Poland,
  August 28 - September 1, 1989, Proceedings}, volume 379 of {\em Lecture Notes
  in Computer Science}, pages 431--444. Springer, 1989.

\bibitem{regnier1998complexity}
Mireille R{\'e}gnier and Wojciech Szpankowski.
\newblock Complexity of sequential pattern matching algorithms.
\newblock In {\em International Workshop on Randomization and Approximation
  Techniques in Computer Science}, pages 187--199. Springer, 1998.

\end{thebibliography}

\newpage
\section{Appendix: proofs}

\stationary*

\begin{proof}
Recall that by construction, after reading $\pref(\Text,j)$ in $\KTX$ (or in $\A_\Pattern$, or in $\MTX$), we reach the state $u$ which is the longest suffix of $\pref(\Text,j)$ that is also a strict prefix of $\Pattern$. Thus, $p_\Pattern(j,u)$ is the probability that $u$ is a suffix of $\pref(\Text,j)$ and that no longer strict prefix of $\Pattern$ is also a suffix of $\pref(\Text,j)$. Let $B_X(u)\subset Q_\Pattern$ be the set of strict prefixes of $X$ that have $u$ has strict border:
\[
B_X(u) = \{v\in Q_X: u\text{ is a strict border of }v\}.
\]
We then have:
\[
p_X(j,u) = \pi(u) - \mathbb{P}\left(\exists v\in B_X(u),\ v \text{ suffix of }\pref(\Text,j)\right).
\]
Observe that if $v \in B_X(u)$ satisfies $u \neq \border{v}$,  
then setting $w := \border{v}$ implies that $w \in B_X(u)$ and 
\[
v\text{ suffix of }\pref(\Text,j) \Rightarrow
w\text{ suffix of }\pref(\Text,j).
\]
Hence 
$\mathbb{P}(v\text{ or }w\text{ is a suffix of }\pref(\Text,j)) = \mathbb{P}(w\text{ is a suffix of }\pref(\Text,j))$. Thus, we can restrict the existential quantifier:
\begin{align*}
\mathbb{P}\big(\exists v\in B_X(u),&\ v \text{ is a suffix of }\pref(\Text,j)\big) \\
& = \mathbb{P}\big(\exists v\in B_X(u),\ \border{v}=u\text{ and }v \text{ is a suffix of }\pref(\Text,j)\big)\\
& = \mathbb{P}\big(\exists v\in Q_X,\ \border{v}=u\text{ and }v \text{ is a suffix of }\pref(\Text,j)\big).
\end{align*}
Finally, let $v,v'\in B_X(u)$ satisfy $\border{v}=\border{v'}=u$. 
If $v\neq v'$ then  both $ v$ and $ v'$ cannot simultaneously be suffixes of $ \pref(\Text, j)$. 
%it is not possible that both $v$ and $v'$ are suffixes of $\pref(\Text,j)$. 
Indeed, assume without loss of generality that $ |v| < |v'|$. Suppose for contradiction that both $ v$ and $ v'$ are suffixes of $ \pref(\Text,j)$.
Then, $v$ is a strict suffix of $v'$, but since both are prefixes of $X$, $v$ is also a prefix of $v'$.  
This implies that $v$ would be a border of $v'$ longer than $u$, contradicting $\border{v'} = u$.
Thus, the events are disjoint for all $v \in Q_X$ such that $\border{v} = u$, leading to:  
\begin{align*}
\mathbb{P}\big(\exists v\in Q_X,&\ \border{v}=u\text{ and }v \text{ is a suffix of }\pref(\Text,j)\big) \\
& = \sum_{\substack{v\in Q_X\\\border{v}=u}} \mathbb{P}(v\text{ is a suffix of }\pref(\Text,j)\big) \\
& = \sum_{\substack{v\in Q_X\\\border{v}=u}}\pi(v).
\end{align*}
This concludes the proof.
\end{proof}

\procomparisons*

\begin{proof}
The transition systems of both $\KTX$ and $\MTX$ are identical to $\A_\Pattern$, as they only involve adding outputs to $\A_\Pattern$. 

For any non-negative integer $j<n$, let $C_{\Pattern,j}$ denote the random variable representing the number of letter comparisons performed by Algorithm~\ref{algo:find} during the $j$-th iteration of the main loop. 
By Lemma~\ref{lm:transducer comparisons}, the number of comparisons is equal to the length of the output $\OutputTransducer_\KTX\left(u\xrightarrow{\alpha}\right)$ for this iteration. 
By the law of total probability, we can partition the possibilities depending on the value of the variable $i$ (which corresponds to the current state $u$ in $\KTX$) and the letter $\anyletter=T_j$. 
Let $\mathds{1}_{u\rightarrow{\anyletter},k}$ denote the indicator function defined by:
\[
\mathds{1}_{u\rightarrow{\anyletter},k} = \begin{cases}
1 & \text{if }\left|\OutputTransducer_\KTX\left(u\xrightarrow{\alpha}\right)\right| = k\\
0 & \text{otherwise.}
\end{cases}
\]

Using the notation $s_\Pattern(\varepsilon,j)$ to denote the state reached after $j$ steps in $\KTX$ (or $\A_\Pattern$), we obtain:
\begin{align*}
\mathbb{P}\left(C_{\Pattern,j} = k\right)
&= \sum_{u\in Q_\Pattern}\sum_{\anyletter\in\alphabet}
\mathds{1}_{u\rightarrow{\anyletter},k}\cdot \mathbb{P}(s_\Pattern(\varepsilon,j)=u)\cdot \mathbb{P}(\Text[j]=\anyletter) \\
& = \sum_{u\in Q_\Pattern}p_X(j,u)\sum_{\anyletter\in\alphabet}\pi(\anyletter)\cdot \mathds{1}_{u\rightarrow{\anyletter},k}.
\end{align*}
From this, we derive:
\[
\mathbb{E}[C_{\Pattern,j}] 
 = \sum_{k\geq 1} k\cdot\mathbb{P}\left(C_{\Pattern,j}=k\right)
 = \sum_{u\in Q_\Pattern}p_\Pattern(j,u)\sum_{\anyletter\in\alphabet}
\pi(\anyletter)\cdot \left|\OutputTransducer_\KTX\left(u\xrightarrow{\alpha}\right)\right|.
\]
Therefore, by Lemma~\ref{lm:stationary probability}, $\mathbb{E}[C_{\Pattern,j}]=C_\Pattern$ for $j\geq m$. 
By linearity of the expectation, the expected total number of comparisons is
$\sum_{j=0}^{n-1}\mathbb{E}[C_{\Pattern,j}]$. This sum is asymptotically equivalent to $\sum_{j=m}^{n-1}\mathbb{E}[C_{\Pattern,j}]$ as $n\rightarrow \infty$, hence to $(n-m)\cdot C_\Pattern$, and therefore to $C_\Pattern\cdot n$. 
Moreover, we have the bound $C_\Pattern\geq 1$, since each character of $\Text$ is read, and $C_\Pattern\leq 2$ as fewer than $2n$ comparisons are performed in both \MP and \KMP~\cite{DBLP:books/ox/CrochemoreR94}.
\end{proof}

\countertypical*

\begin{proof}
Since $\Pattern$ contains at least two distinct letters, it cannot be a suffix of both $\pref(\Text,j)$ and $\pref(\Text,j+1)$. 
Hence, the sequence $s$ of taken/not-taken outcomes for the test $\Pattern[i]\neq \Text[j]$ follows the pattern $(N^+T)^*N^*$. 
This implies that every step to the right in the local predictor (i.e. for every $T$ in sequence~$s$) is immediately followed by a step to the left, except possibly for the last step.
Thus, if the local predictor reaches the state $\sNN$, it remains in $\sN$ or $\sNN$ forever, always predicting that the branch is not taken. 

Let $\alpha$ be the last letter of $\Pattern$ and define $F(\Text)$ as the smallest index $j\geq 2$ in $\Text$ such that $\Text_{j-2}$, $\Text_{j-1}$ and $\Text_j$ are all different from $\anyletter$, with $F(\Text)=+\infty$ if no such index exists. 
Let $q=1-\pi(\anyletter)$. 
For $\ell=3\lceil \log_x n \rceil$, with $x=1/(1-q^3)$, we have:
\[
\mathbb{P}\left(F(\Text) \geq \ell\right)
\leq \mathbb{P}\left(\forall i\in\{1,\ldots,\ell/3\},\ \Text_{3i-3}\neq\anyletter\text{ or }\Text_{3i-2}\neq\anyletter\text{ or }\Text_{3i-1}\neq\anyletter\right). 
\]
By independence of letters in $\Text$, it follows that:
\[
\mathbb{P}\left(F(\Text) \geq \ell\right)
\leq (1-q^3)^{\ell/3}
\leq (1-q^3)^{\log_x n} = n^{\log_x(1-q^3)} = \frac1n.
\]
Let $M(\Text,\lambda)$ denote the number of mispredictions induced by the counter update when applying Algorithm~\ref{algo:find} to a random text~$\Text$ of length $n$, with the predictor starting in state $\lambda\in\{\sNN,\sN,\sT,\sTT\}$. 
For sufficiently large $n$, $\mathbb{P}(F(\Text)<\ell)>0$ and $\mathbb{P}(F(\Text)\geq\ell)>0$, leading to:
\[
\mathbb{E}[M(\Text,\lambda)] = \mathbb{E}[M(\Text,\lambda)\mid F(\Text)<\ell]\cdot \mathbb{P}(F(\Text)<\ell)+\mathbb{E}[M(\Text,\lambda)\mid F(\Text)\geq\ell]\cdot \mathbb{P}(F(\Text)\geq\ell).
\]

For the first term, in the worst case, each counter update (i.e., each iteration of the main loop) results in at most one misprediction, giving:
\[
\mathbb{E}[M(\Text,\lambda)\mid F(\Text)\geq\ell]\cdot \mathbb{P}(F(\Text)\geq\ell) \leq n\cdot \frac1n = 1.
\]

Let $\suff(w,k)$ denote the suffix of length $k\leq |w|$ of $|w|$. 
Additionally, let $\text{Occ}(\Text,\Pattern)$ denote the number of occurrences of $\Pattern$ in $\Text$.
If $\Text$ is such that $F(\Text)<\infty$, then for all $\lambda\in\{\sNN,\sN,\sT,\sTT\}$ we have:
\begin{equation}\label{eq:FW}
\left|M(\Text,\lambda)-M(\suff(\Text,n-F(\Text)),\sNN)\right| \leq F(\Text),
\end{equation}
 since this predictor produces at most one misprediction per letter of $W$.
 Moreover, as $\Text$ is generated by a memoryless source, 
 is memoryless, $\suff(\Text,n-F(\Text))$ is itself a random word of length
 $n-F(\Text)$ , produced by a memoryless source with the same probability distribution $\pi$ over $\alphabet$. 
 As mentioned above, when the predictor starts $\sNN$, exactly one misprediction occurs per occurrence of $\Pattern$:
 \[
 \mathbb{E}[M(\suff(\Text,n-F(\Text)),\sNN)]
 = \mathbb{E}[\text{Occ}(\suff(\Text,n-F(\Text)),\Pattern)
 = (n-F(\Text)) \pi(X).
 \]
 Therefore,  using Eq.~\eqref{eq:FW}, we obtain:
 \[
 \mathbb{E}[M(\Text,\lambda)\mid F(\Text)<\ell] = n\pi(\Pattern) +\O(\ell).
 \]
 This concludes the proof, as $\mathbb{P}(F(w) < \ell) = 1+o(1)$ and $\ell = \O(\log n)$.
\end{proof}

\misspredam*

\begin{proof}
We consider the different cases separately.

\smallskip\noindent
$\circ$ \underline{$X=\anyletter$:} for each letter $\Text[j]$ of $\Text$, the branch from Line~\ref{line:find-counter-if} is taken if and only if $\Text[j]=\anyletter$. 
Thus, the predictor of \cref{fig:predictor} can be directly modeled as a Markov chain, where transitions occur with probability $\pi(\anyletter)$ to the right and $1-\pi(\anyletter)$ to the left as illustrated in the following figure.

\begin{center}
    \begin{tikzpicture}
        \node[draw,circle] (NN) at (0,0) {$\sNN$};
        \node[draw,circle] (N) at (2,0) {$\sN$};
        \node[draw,circle] (T) at (4,0) {$\sT$};
        \node[draw,circle] (TT) at (6,0) {$\sTT$};

        \draw (NN) edge[> = stealth',loop left] node {\small $1-\pi(\anyletter)$} (NN);
        \draw (N) edge[> = stealth',->,bend left] node[below] {\small $1-\pi(\anyletter)$} (NN);
        \draw (T) edge[> = stealth',->,bend left, very thick] node[below] {\small $1-\pi(\anyletter)$} (N);
        \draw (TT) edge[> = stealth',->,bend left, very thick] node[below] {\small $1-\pi(\anyletter)$} (T);
       
        \draw (NN) edge[> = stealth',->,bend left, very thick] node[above] {\small $\pi(\anyletter)$} (N);
        \draw (N) edge[> = stealth',->,bend left, very thick] node[above] {\small $\pi(\anyletter)$} (T);
        \draw (T) edge[> = stealth',->,bend left] node[above] {\small $\pi(\anyletter)$} (TT);
        \draw (TT) edge[> = stealth',loop right] node {\small $\pi(\anyletter)$} (TT);
    \end{tikzpicture}
\end{center}
The stationary distribution $\pi_0$  of this Markov chain can be readily computed. 
From this, the expected number of mispredictions corresponds to the expected number of times a bold edge is taken over $n$ steps,  which is asymptotically equal to
\[
n\times\sum_{\lambda\in\{\sNN,\sN,\sT,\sTT\}} \pi_0(\lambda) \sum_{\substack{\lambda'\in\{\sNN,\sN,\sT,\sTT\}\\\lambda\rightarrow\lambda'\text{ mispr.}}} \mathbb{P}(\lambda\rightarrow\lambda').
\]
After straightforward computations, this quantity is equal to $\kappa_1(\pi(\anyletter))$. 
Observe that this result is identical to applying $n$ independent Bernoulli trials with parameter $\pi(\anyletter)$ to the local predictor.
This scenario has been extensively studied in the literature, for instance in~\cite{AuNiPi2016,MaNeWi15}.

\smallskip\noindent
$\circ$ \underline{$X=\anyletter\anyletter$:} the proof sketch in the article covers most of the details. 
If the predictor's initial state is $\sTT$, then after reading the first factor of the form $\anyletter^*\overline\anyletter$, it is no longer in state $\sTT$ anymore, since every such factor ends with a not-taken branch.
Moreover, it is exponentially unlikely that $\Text = \anyletter^n$. 
Viewing a random text $\Text$ as a sequence of factors $\anyletter^*\overline\anyletter$, we approximate it with a word $\widetilde\Text$ by independently drawing  $\ell:=\lceil (1-\pi(\anyletter))n\rceil$ such factors at random.
Recall that the negative binomial law $\text{NB}(k,p)$ models the number of failures before obtaining $k$ success in a sequence of independent Bernoulli trials with success probability $p$.
If a success corresponds to drawing $\overline\anyletter$, which occurs with probability $1-\pi(\anyletter)$, then $|\widetilde\Text|-\ell$ follows a negative binomial law with parameters $\ell$ and $1-\pi(\anyletter)$.
Using classical results:
\[
\mathbb{E}[\text{NB}(k,p)] = \frac{k(1-p)}{p}\text{ and }\mathbb{V}[\text{NB}(k,p)] = \frac{k(1-p)}{p^2}.
\]
Thus for fixed $p$, the distribution is concentrated around its expected value by Chebyshev's inequality. 
In our case
$\mathbb{E}[|\widetilde\Text|-\ell]= \frac{\ell\pi(\anyletter)}{1-\pi(\anyletter)}$, which gives:
\[
\mathbb{E}[|\widetilde\Text|]=\frac{\ell}{1-\pi(\anyletter)} = n + \O(1).
\]
Together with  concentration around the mean, this shows that $\widetilde\Text$ differs from a random word of length $n$ only by a sublinear number of letters. 
Since each letter contributes to at most one misprediction for this branch, the main asymptotic term can be computed directly on $\widetilde\Text$, concluding the proof for $X=\anyletter\anyletter$.

\smallskip\noindent
$\circ$ \underline{$X=\anyletter^m$ with $m\geq 3$:} as we just proved, we can analyze the word $\widetilde\Text$ obtained by generating $\ell$ factors $\anyletter^*\overline{\anyletter}$. 
Since reading such a factor always ends with a not-taken branch, the predictor can be in state $\sTT$ only at its initial configuration and at the very end. 
Thus, we focus on the three remaining states to compute the evolution of the predictor and the number of mispredictions for each factor $\anyletter^k\overline{\anyletter}$:
\[
    \begin{array}{c|l|l|l|l|l|l|l|l|l|l}
       & \multicolumn{2}{c|}{k=0} &  \multicolumn{2}{c|}{k=1,2, 3\text{ to }m\!-\!1}&  \multicolumn{2}{c|}{k=m}&  \multicolumn{2}{c|}{k= m+1}&  \multicolumn{2}{c}{k\geq m+2}\\
        & \multicolumn{2}{c|}{N} & \multicolumn{2}{c|}{N^{k+1}} & \multicolumn{2}{c|}{N^{m-1}TN} & \multicolumn{2}{c|}{N^{m-1}TTN} & \multicolumn{2}{c}{N^{m-1}T^{k-m+1}N}\\ 
       \hline
       \sNN & 
        \rightarrow\sNN & 0 \text{ m.}&
        \quad~\rightarrow \sNN  & 0 \text{ m.}&
        \rightarrow \sNN  & 1 \text{ m.}&
        \rightarrow \sN  & 3 \text{ m.}&
        \quad\rightarrow \sT & 3 \text{ m.}\\
       \sN &
        \rightarrow\sNN & 0 \text{ m.}&
        \quad~\rightarrow \sNN  & 0 \text{ m.}&
        \rightarrow \sNN  & 1 \text{ m.}&
        \rightarrow \sN  & 3 \text{ m.}&
        \quad\rightarrow \sT & 3 \text{ m.}\\
       \sT &
        \rightarrow\sN & 1 \text{ m.}&
        \quad~\rightarrow \sNN  & 1 \text{ m.}&
        \rightarrow \sNN  & 2 \text{ m.}&
        \rightarrow \sN  & 4 \text{ m.}&
        \quad\rightarrow \sT & 4 \text{ m.}\\
    \end{array}
\]

Observe that, as in the case $X=\anyletter\anyletter$, states $\sNN$ and $\sN$ exhibit identical behavior and can be merged into a single state  $\tilde\sN$ for analysis. 
After this reduction, the state $\sT$ has the same transitions but induces exactly one additional misprediction.
Moreover, the probability that the predictor is in state $\sT$ after reading a factor $\beta_i$ is the probability that $\beta_i$ contains at least  $m+2$ occurrences of $\anyletter$, which is $\pi(\anyletter)^{m+2}$. 
Hence the probability that the state is $\sNN$ or $\sN$ is $1-\pi(\anyletter)^{m+2}$. 
If the predictor is on $\sNN$ or~$\sN$, it produces 1 misprediction with probability $\pi(\anyletter)^m(1-\pi(\anyletter))$ and 3 mispredictions with probability $\pi(\anyletter)^{m+1}$. 
Thus the average number of mispredictions induced by the next factor, starting from one of these states is
\begin{align*}  
\mathbb{E}[\text{misp. next factor}\mid \text{state}=\tilde\sN]& = \pi(\anyletter)^m(1-\pi(\anyletter))+3 \pi(\anyletter)^{m+1}  = \pi(\anyletter)^m\left(1+2\pi(\anyletter)\right).
\end{align*}
And similarly,
\[
\mathbb{E}[\text{misp. next factor}\mid \text{state}=\sT]= 1+ \pi(\anyletter)^m\left(1+2\pi(\anyletter)\right).
\]

Hence, the expected number of mispredictions for $\widetilde\Text$ is asymptotically:
\[
\ell\times\left(\pi(\anyletter)^{m} (1+2\pi(\anyletter)) + \pi(\anyletter)^{m+2}\right) = \ell\times \pi(\anyletter)^{m} (1+\pi(\anyletter))^2\text{.}
\]
Since $\ell=\lceil (1-\pi(\anyletter)n\rceil$ is asymptotically equivalent to $(1-\pi(\anyletter))n$, this completes the proof.
\end{proof}

\inonnegative*

\begin{proof}
For the first part, concerning the error term, we follow the reasoning outlined in the main part of the article.
In every state, there exists at least one letter $\anyletter$ such that the outgoing transition in $\WKTX$ (or $\WMTX$) produces $T$.
Since other transitions cannot move the predictor toward $\sNN$ before reaching $\sTT$, with probability at least
$\pi_{\min} := \min_{\anyletter\in\alphabet}\pi(\anyletter)$
the predictor transitions toward $\sTT$; otherwise, it remains in the same state. 
The probability that fewer than three transitions toward $\sTT$ occur within the first $\ell:=\lceil\log_x n\rceil$ letters of $\Text$, where $x=(1-\pi)^{-1/2}$, , is upper bounded by
The probability that there are less than three moves toward $\sTT$ when reading the first $\ell:=\lceil\log_x n\rceil$ letters of $\Text$, with $x=(1-\pi)^{-1/2}$, is therefore upper bounded by 
(we upper bounded the probability of such moves by $1$ and the probability of staying in the same state by $1-\pi_{\min}$):
\[
\left(1-\pi_{\min}\right)^\ell + \ell\left(1-\pi_{\min}\right)^{\ell-1}
+ \frac{\ell(\ell-1)}2\left(1-\pi_{\min}\right)^{\ell-2}
= \O\left(\ell^2\left(1-\pi_{\min}\right)^{\ell}\right)
= \O\left(\frac{\log^2 n}{n^2}\right).
\]
Thus, the contribution of words $\Text$ for which $\sTT$ is not reached within the first $\lceil\log_x n\rceil$ letters is negligible.

The remainder is detailed in the main part of the article, with the formula for $G_X$ following as a consequence of the classical Ergodic Theorem~\cite{LePeWe08}, where the stationary distribution is given by~\cref{lm:stationary probability}.
\end{proof}

\newpage
\paragraph*{Additional figures}

In the first series of four figures below, we used the formulas obtained in the article to compute the expected number of mispredictions for each branch, as well as the total number of mispredictions. 
The results are represented using text symbols as $\pi(a)$ varies. 
The \MP algorithm is depicted with dashed lines. 
We ran our code for the patterns $aaaa$,$aaab$,$abab$, and$abbb$. Each plot is generated within a few seconds on a personal laptop.

\begin{figure}[h]
    \centering
    %\vspace{-.8cm}
    \includegraphics[scale=.4]{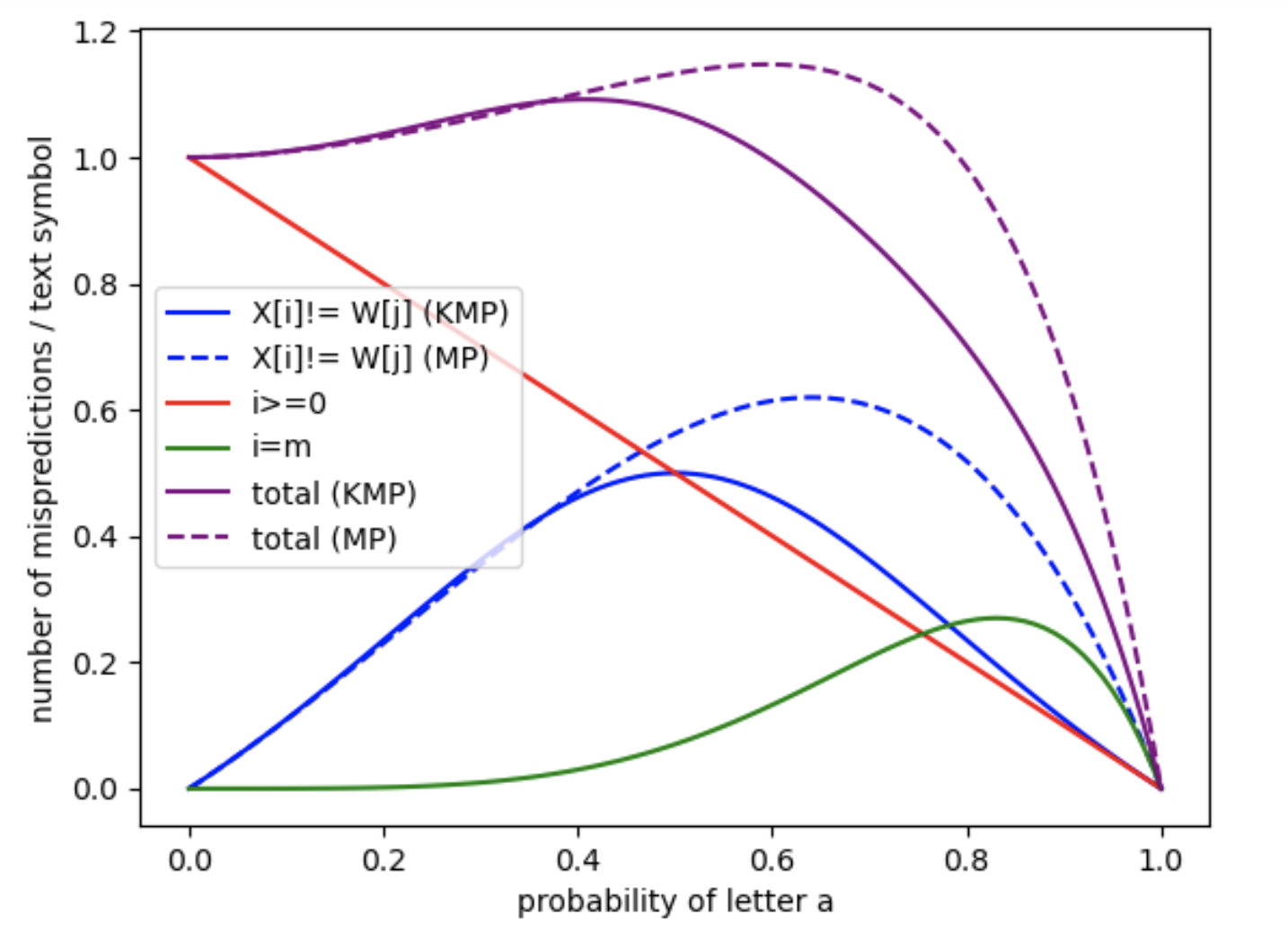}
    \caption{\label{fig:comparatif}%
        Pattern $X=aaaa$.
    }
\end{figure}

\begin{figure}[h]
    \centering
    %\vspace{-.8cm}
    \includegraphics[scale=.4]{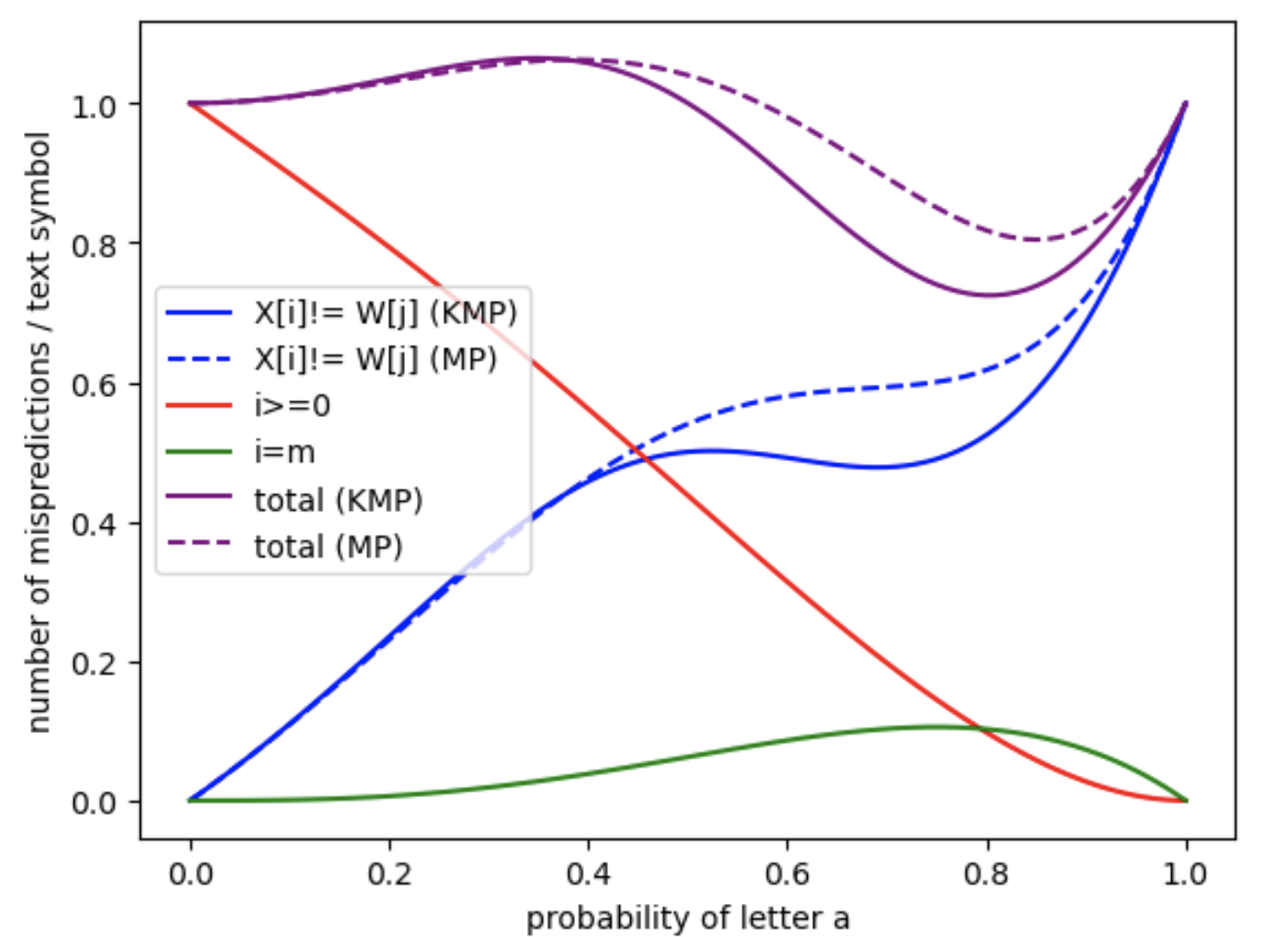}
    \caption{\label{fig:comparatif}%
        Pattern $X=aaab$.
    }
\end{figure}

\begin{figure}[h]
    \centering
    %\vspace{-.8cm}
    \includegraphics[scale=.4]{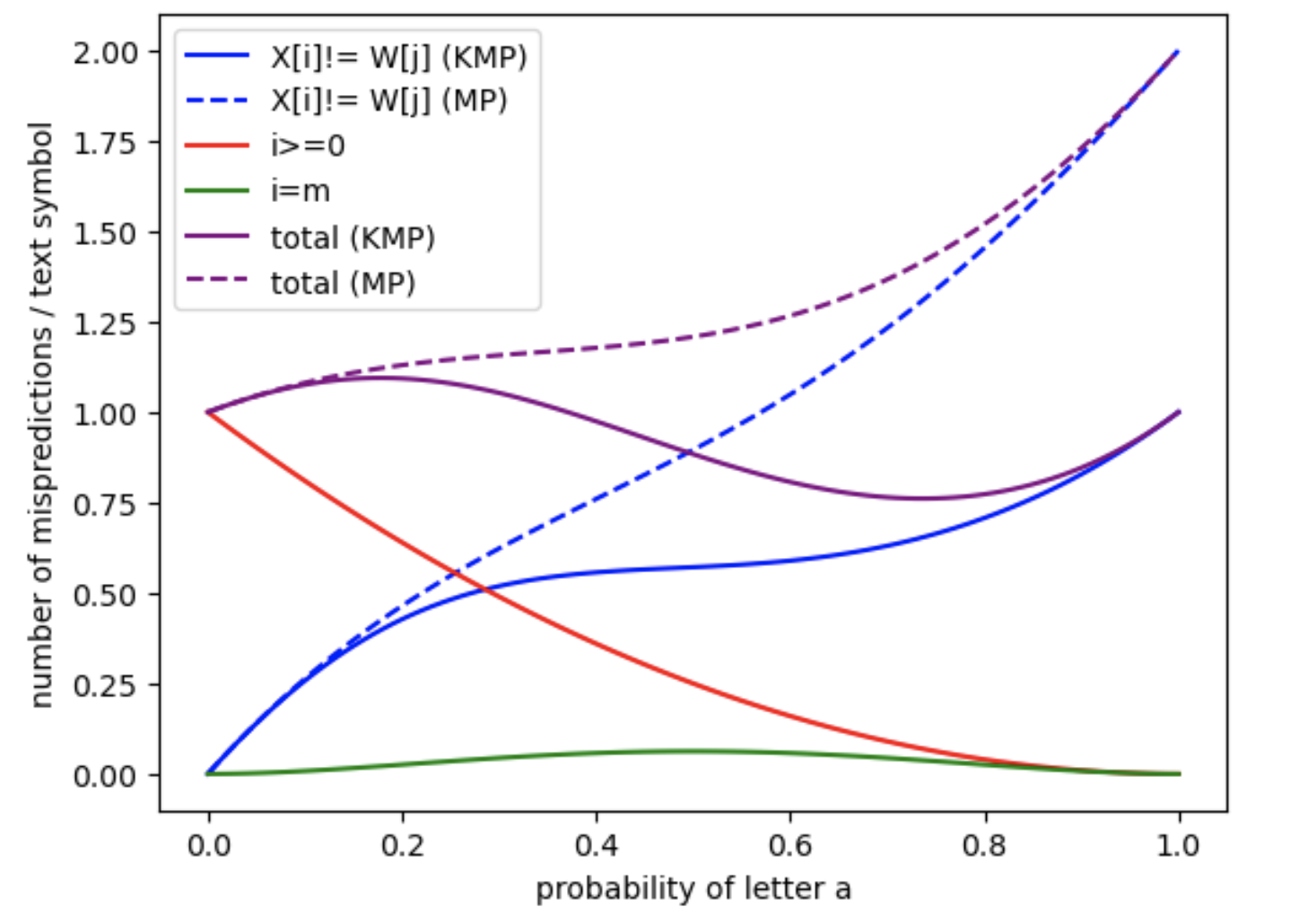}
    \caption{\label{fig:comparatif}%
        Pattern $X=abab$.
    }
\end{figure}

\begin{figure}[h]
    \centering
    %\vspace{-.8cm}
    \includegraphics[scale=.4]{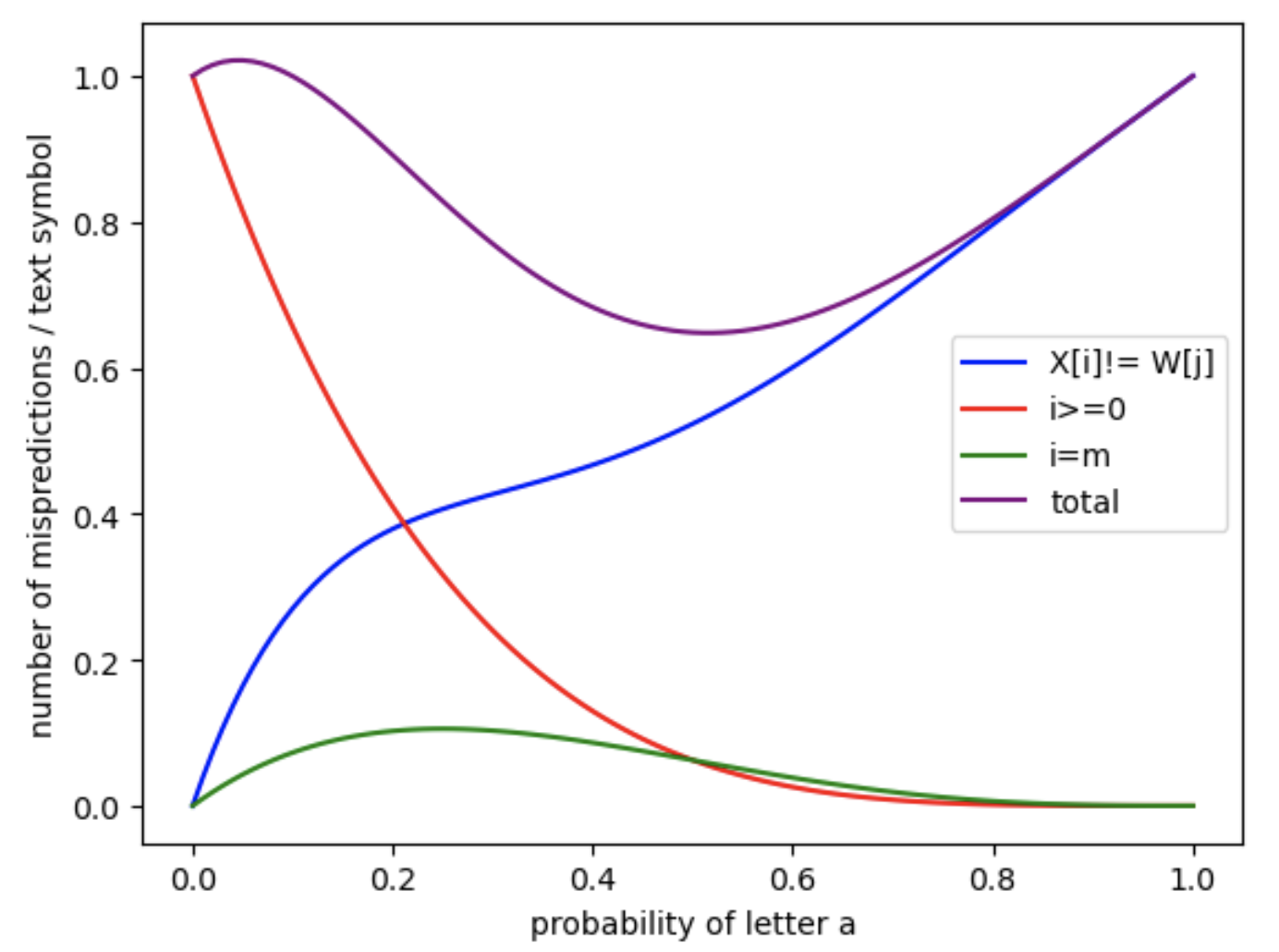}
    \caption{\label{fig:comparatif}%
        Pattern $X=abbb$. Both variant MP and KMP have the same border table.
    }
\end{figure}

\newpage

In the last two figures, we consider all  prefixes of length at least $2$ of $abababb$ and compute the variation of the expected total number of mispredictions per text symbol, for both the \MP and \KMP algorithms. 
The results suggest a form of convergence as the length of the prefix increases, which is expected since reaching the rightmost states of $\mathcal{A}_\Pattern$ becomes increasingly unlikely.
%as it is more and more unlikely to reach the rightmost states of $\A_\Pattern$.

\begin{figure}[h]
    \centering
    %\vspace{-.8cm}
    \includegraphics[scale=.4]{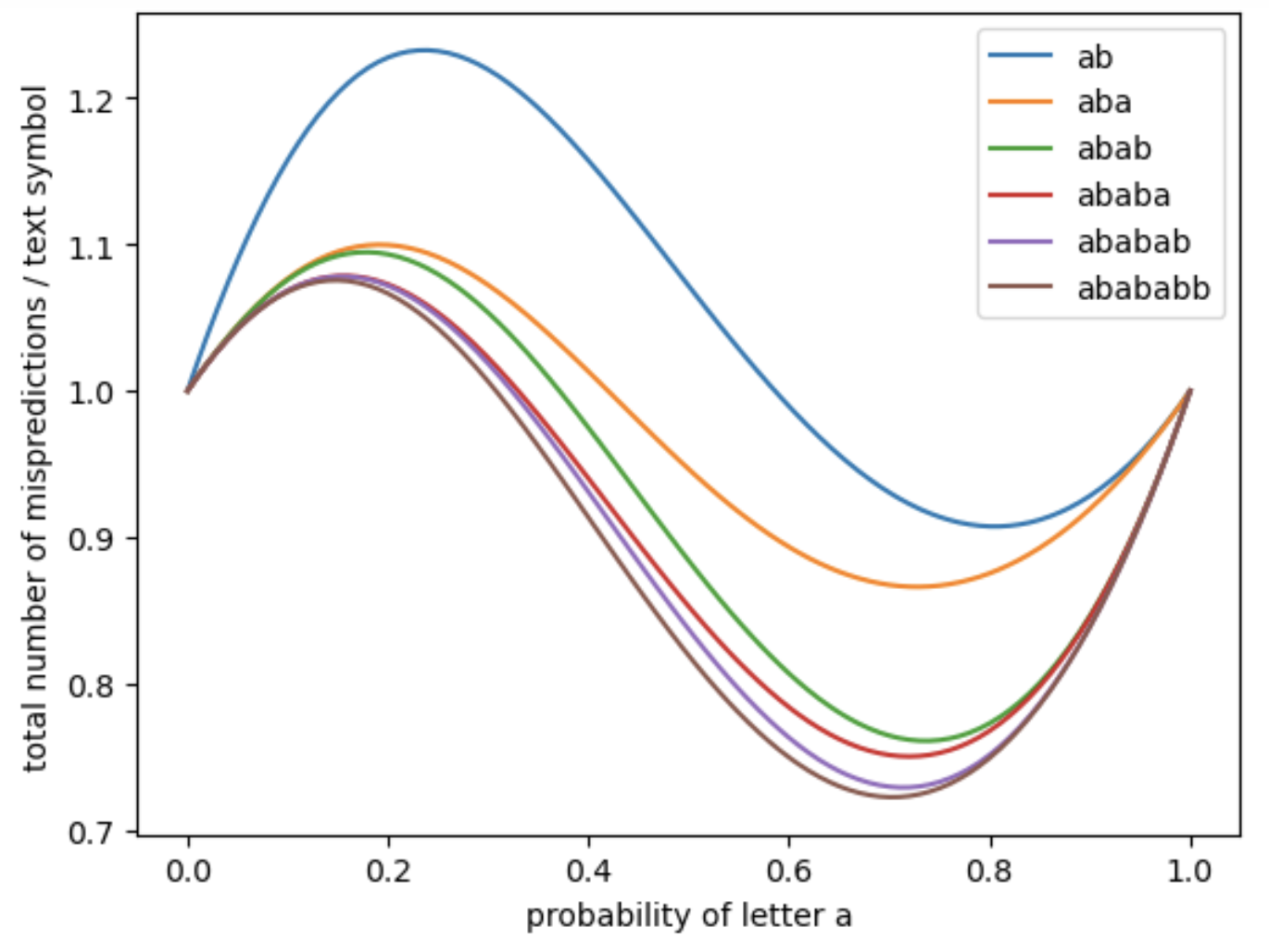}
    \caption{\label{fig:comparatif}%
        Total number of mispredictions for the prefixes of $abababb$ (KMP).
    }
\end{figure}

\begin{figure}[h]
    \centering
    %\vspace{-.8cm}
    \includegraphics[scale=.4]{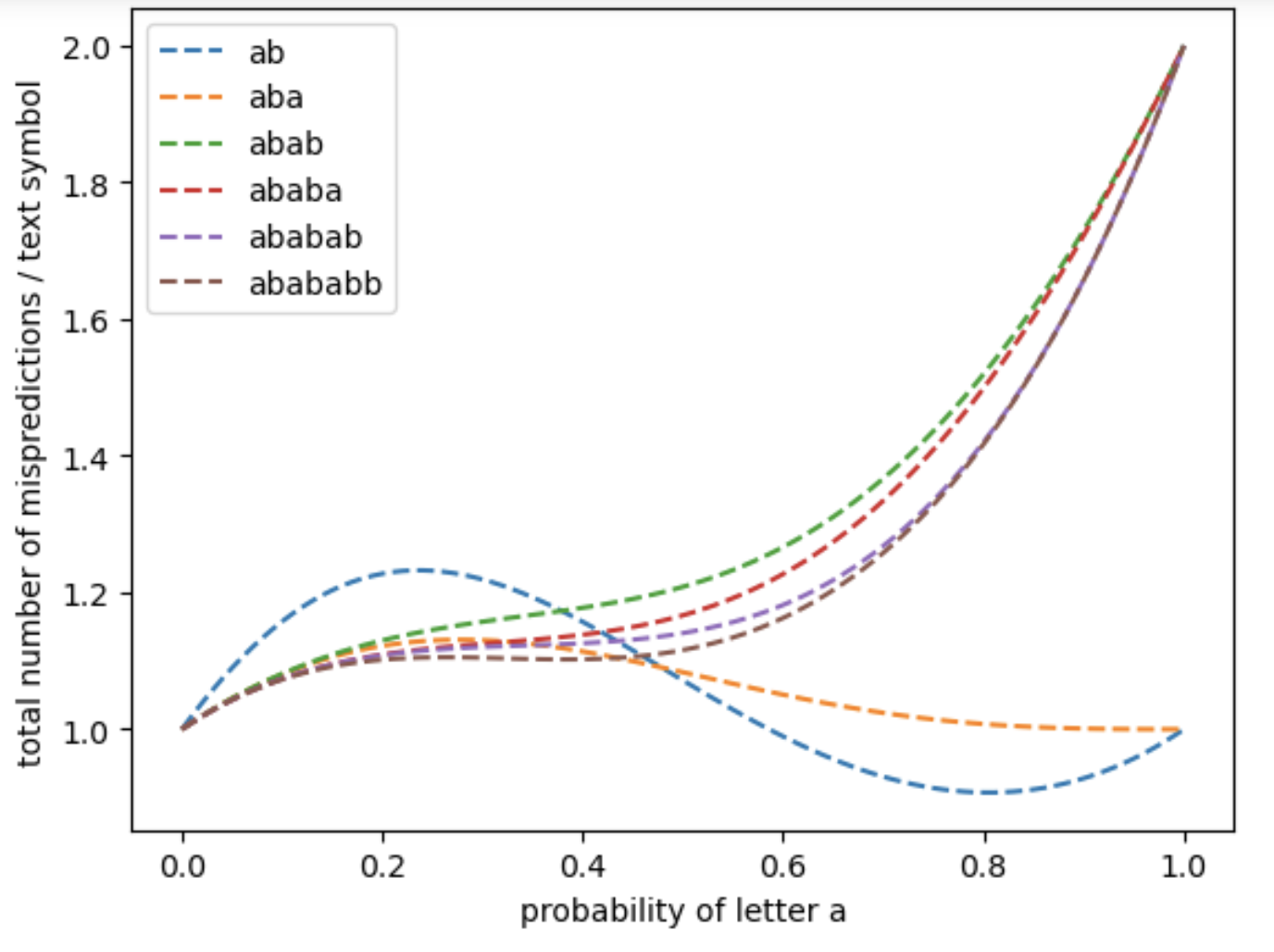}
    \caption{\label{fig:comparatif}%
        Total number of mispredictions for the prefixes of $abababb$ (MP).
    }
\end{figure}

\end{document}